\documentclass[10pt, conference, letterpaper]{IEEEtran}
\def\thisismainpaper{0}   

\usepackage{graphicx}
\usepackage{amsmath}
\usepackage{blkarray,booktabs,bigstrut} 
\usepackage{multirow}
\usepackage{epsfig}
\usepackage{mathrsfs}
\usepackage{amssymb}
\usepackage{comment}
\usepackage{array}
\usepackage{amsthm}
\usepackage{blkarray}
\usepackage{enumerate}
\usepackage{url}
\usepackage{hyperref}
\usepackage{chemarrow}
\usepackage{color}
\usepackage{xcolor}
\usepackage{bbm}
\usepackage[lined,ruled,linesnumbered,noend]{algorithm2e}
\makeatletter
\newcommand{\nosemic}{\renewcommand{\@endalgocfline}{\relax}}
\newcommand{\dosemic}{\renewcommand{\@endalgocfline}{\algocf@endline}}

\let\oldnl\nl
\newcommand{\nonl}{\renewcommand{\nl}{\let\nl\oldnl}}
\makeatother

\usepackage{epsf,psfrag}
\usepackage{epsfig}
\usepackage{bm}
\usepackage{mathtools}
\usepackage{multicol,lipsum} 
\usepackage[tight,footnotesize]{subfigure}

\theoremstyle{definition}
\newtheorem{theorem}{Theorem}[section]
\newtheorem{lemma}[theorem]{Lemma}

\newtheorem{definition}{Definition}

\DeclareMathOperator{\argmin}{\arg\min}

\newcommand{\E}{{\rm I\kern-.3em E}}

\newcommand{\conv}{\mathrm{conv}}
\newcommand{\diam}{\mathrm{diam}}

\def\diag{\operatorname{diag}}
\def\uG{\underline{G}}
\def\uV{\underline{V}}
\def\uE{\underline{E}}
\def\up{\underline{p}}
\def\ue{\underline{e}}

\def\umax{u_{\small \max}}
\def\vmax{v_{\small \max}}
\def\Cmin{C_{\small \min}}

\def\ddefloop#1{\ifx\ddefloop#1\else\ddef{#1}\expandafter\ddefloop\fi}
\def\ddef#1{\expandafter\def\csname bb#1\endcsname{\ensuremath{\mathbb{#1}}}}
\ddefloop ABCDEFGHIJKLMNOPQRSTUVWXYZ\ddefloop
\def\ddef#1{\expandafter\def\csname c#1\endcsname{\ensuremath{\mathcal{#1}}}}
\ddefloop ABCDEFGHIJKLMNOPQRSTUVWXYZ\ddefloop
\def\ddef#1{\expandafter\def\csname v#1\endcsname{\ensuremath{\boldsymbol{#1}}}}
\ddefloop ABCDEFGHIJKLMNOPQRSTUVWXYZabcdefghijklmnopqrstuvwxyz\ddefloop
\def\ddef#1{\expandafter\def\csname u#1\endcsname{\ensuremath{\underline{#1}}}}
\ddefloop ABCDEFGHIJKLMNOPQRSTUVWXYZabcdefghijklmnopqrstuvwxyz\ddefloop
\def\ddef#1{\expandafter\def\csname v#1\endcsname{\ensuremath{\boldsymbol{\csname #1\endcsname}}}}
\ddefloop {alpha}{beta}{gamma}{delta}{epsilon}{varepsilon}{zeta}{eta}{theta}{vartheta}{iota}{kappa}{lambda}{mu}{nu}{xi}{pi}{varpi}{rho}{varrho}{sigma}{varsigma}{tau}{upsilon}{phi}{varphi}{chi}{psi}{omega}{Gamma}{Delta}{Theta}{Lambda}{Xi}{Pi}{Sigma}{varSigma}{Upsilon}{Phi}{Psi}{Omega}{ell}\ddefloop

\ifCLASSINFOpdf

\else

\fi

\begin{document}

\IEEEoverridecommandlockouts
%
\title{Communication Optimization for Decentralized Learning atop Bandwidth-limited Edge Networks
}
\author{
\IEEEauthorblockN{Tingyang Sun, Tuan Nguyen, and Ting He} 

\IEEEauthorblockA{
Pennsylvania State University, University Park, PA, USA. Email: \{tfs5679,tmn5319,tinghe\}@psu.edu
}
\thanks{This work was supported by the National Science Foundation under award CNS-2106294.}
}

\maketitle

\begin{abstract}
Decentralized federated learning (DFL) is a promising machine learning paradigm for bringing artificial intelligence (AI) capabilities to the network edge. Running DFL on top of edge networks, however, faces severe performance challenges due to the extensive parameter exchanges between agents. 
 Most existing solutions for these challenges were based on simplistic communication models, which cannot capture the case of learning over a multi-hop bandwidth-limited network. In this work, we address this problem by jointly designing the communication scheme for the overlay network formed by the agents and the mixing matrix that controls the communication demands between the agents. By carefully analyzing the properties of our problem, we cast each design problem into a tractable optimization and develop an efficient algorithm with guaranteed performance. Our evaluations based on real topology and data show that the proposed algorithm can  reduce the total training time by over $80\%$ compared to the baseline without sacrificing accuracy, while significantly improving the computational efficiency over the state of the art. 
 \looseness=-1
\end{abstract}

\begin{IEEEkeywords}
Decentralized federated learning, overlay network, mixing matrix design. 
\end{IEEEkeywords}

\section{Introduction}\label{sec:Introduction}

\emph{Decentralized federated learning (DFL)}~\cite{Lian17NIPS} is an emerging machine learning paradigm that allows multiple learning agents to collaboratively learn a shared model from their local data without directly sharing the data. In contrast to the centralized federated learning (FL) paradigm~\cite{McMahan17AISTATS}, DFL gets rid of parameter servers by letting the learning agents directly exchange model updates with their neighbors through peer-to-peer connections, which are then aggregated locally~\cite{Kairouz21book}. Since its introduction, DFL has attracted significant attention due to its robustness against a single point of failure and ability to balance the communication complexity across nodes without increasing the computational complexity~\cite{Lian17NIPS}.

Meanwhile, DFL still faces significant performance challenges due to the extensive data transfer between agents. As in FL, the agents need to communicate repeatedly to exchange local model updates until reaching global convergence, which often incurs a substantial communication cost due to the large model size. Such communication cost can dominate the total cost of the learning task, 
and the problem is exacerbated in edge networks that are more bandwidth-limited than datacenter networks~\cite{chen2022federated}. 
Tremendous efforts have been devoted to reducing the communication cost, including model compression such as~\cite{Compression1}, 
optimization of communication-controlling hyperparameters such as ~\cite{Chiu23JSAC}, 
and adaptive communication such as~\cite{Singh21JSAIT}. 
These approaches are compatible with each other and thus can be combined. In this work, we will focus on the optimization of a particular type of hyperparameter called \emph{mixing matrix} that controls the communications in DFL. 

\begin{figure}[t!]
\vspace{-.0em}
\centerline{\mbox{\includegraphics[width=.7\linewidth]{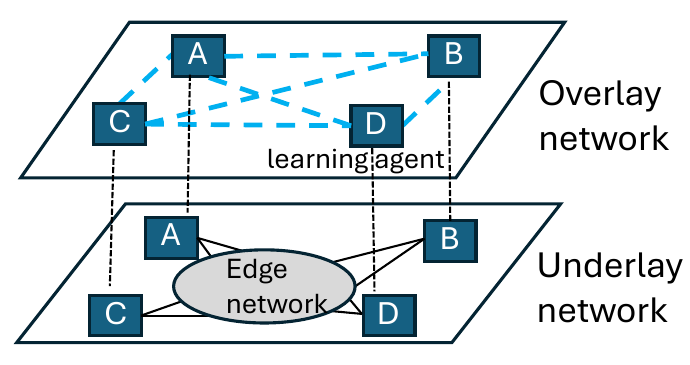}}}
\vspace{-1em}
\caption{\small Overlay-based DFL atop edge network. }
\label{fig:overlay_underlay}
\vspace{-.5em}
\end{figure}

As explained later in Section~\ref{subsec:Model of Decentralized Learning}, only agent pairs corresponding non-zero entries in the mixing matrix need to communicate during DFL, allowing the mixing matrix to control the communication demands between agents. This observation has triggered a series of studies on how to optimally design the mixing matrix such as \cite{MATCHA19,Chiu23JSAC,hua2022efficient,le2023refined}. However, most existing works made the simplistic assumption that each pair of \emph{logically adjacent} agents are also \emph{physically adjacent} in the underlying communication network. 
This assumption has led to simplistic models of the communication cost under a given design, e.g., the maximum degree~\cite{hua2022efficient,le2023refined}, the minimum number of matchings~\cite{MATCHA19}, or other functions of the communication graph~\cite{Chiu23JSAC}. The underlying assumption of all these cost models is that the communication cost (e.g., time) for a given agent-to-agent communication graph only depends on the topology of this graph. 
However, this assumption may not hold in practice, as agents typically communicate through an underlying communication network and the cost for a given set of communications also depends on the state (e.g., topology, routing, link capacities) of this network. In this work, we study the communication optimization for such \emph{overlay-based DFL} over bandwidth-limited edge networks.

As illustrated in Fig.~\ref{fig:overlay_underlay}, in overlay-based DFL, the agents form an \emph{overlay network} with logical connections that indicate which pairs of agents are allowed to communicate, and the underlying communication network serves as an \emph{underlay network} that connects the agents through possibly multi-hop paths. This scenario fundamentally differs from the scenarios studied in previous works in that seemingly disjoint links in the overlay may map to routing paths in the underlay that share (underlay) links. Ignoring such link sharing can cause incorrect prediction of the communication cost under a given design and thus suboptimal designs. For example, the overlay links $(A,B)$ and $(C,D)$ in Fig.~\ref{fig:overlay_underlay} may seem disjoint with a capacity of $C$ each, but concurrently sending two messages of size $\kappa$ over them can take more than $\kappa/C$ time if the corresponding routing paths share the same bottleneck link in the underlay. 

In this work, we study \emph{underlay-aware} communication design for overlay-based DFL in the context of edge networks. Federated learning has been applied in many edge networks such as HetNets~\cite{chen2022federated}, device-to-device networks~\cite{xing2021federated}, IoT networks~\cite{pinyoanuntapong2022toward}, underwater networks~\cite{pei2023fed}, and power line communication networks~\cite{jia2022dispatching}. Compared to other network environments such as inter-datacenter networks \cite{marfoq2020throughput}, edge networks have unique characteristics such as low bandwidth and low propagation delay that lead to different design considerations as explained later (Section~\ref{subsec:Communication Schedule Optimization}). We will build upon our recent discoveries in network tomography~\cite{Huang23MobiHoc} and mixing matrix design~\cite{Chiu23JSAC,Huang24MobiHoc} to develop an overlay-based solution that only requires the participation of the learning agents, making our solution easily deployable in public networks.

\subsection{Related Work}\label{subsec:Related Work}

\textbf{Decentralized federated learning.} Initially proposed under a centralized architecture \cite{McMahan17AISTATS}, FL was later extended to a fully decentralized architecture~\cite{Lian17NIPS}, which was shown to achieve the same computational complexity but a lower communication complexity. 
Since then a number of improvements have been developed such as \cite{ICMLhonor}, but these works only focused on reducing the number of iterations. \looseness=-1
%

\textbf{Communication cost reduction.} 
There are two general approaches for reducing the communication cost in FL: reducing the amount of data per communication through compression, e.g., \cite{Compression1}, 
and reducing the number of communications, e.g.,  \cite{sysml19,Wang19JSAC}. 
The two approaches can be combined for further improvement \cite{Singh20CDC,Singh21JSAIT}. 
Instead of either activating all the links or activating none, it has been recognized that better efficiency can be achieved by activating subsets of links. To this end, \cite{Singh20CDC,Singh21JSAIT} proposed an event-triggered mechanism and \cite{MATCHA19,Chiu23JSAC} proposed to activate links with predetermined probabilities. 
In this regard, our work designs predetermined link activation as in \cite{MATCHA19,Chiu23JSAC}, which provides more predictable performance than event-triggered mechanisms, but {we consider a cost model tailored to overlay-based DFL}: instead of measuring the communication time by the number of matchings  \cite{MATCHA19,Chiu23JSAC} or the maximum degree  \cite{hua2022efficient,le2023refined}, we use the actual time to complete the activated agent-to-agent communications over a bandwidth-limited underlay with possibly shared links. 

\textbf{Topology design in DFL.} The logical topology defining the neighborhoods of learning agents is an important design parameter in DFL. The impact of this topology on the convergence rate of DFL has been mostly captured through the {spectral gap} of the mixing matrix~\cite{Lian17NIPS,Nedic18IEEE,Neglia19INFOCOM,neglia2020decentralized,jiang2023joint} or equivalent parameters~\cite{MATCHA19}. 
Although recent works have identified other parameters that can impact the convergence rate, such as the effective number of neighbors~\cite{vogels2022beyond} and the neighborhood heterogeneity~\cite{le2023refined}, these results just pointed out additional factors and did not invalidate the impact of spectral gap. 
Based on the identified convergence parameters, several solutions have been proposed to design the logical topology to balance the convergence rate and the cost per communication round~\cite{MATCHA19,Chiu23JSAC,le2023refined}, and some solutions combined topology design with other optimizations (e.g., bandwidth allocation~\cite{Wang22Networking}, model pruning~\cite{jiang2023joint}) for further improvement. 
Our work also includes topology design based on a parameter related to the spectral gap, but we consider the joint optimization of the topology and the routing within the overlay in an overlay-underlay network. 

\textbf{Overlay-based DFL.} 
To our knowledge, the only existing works addressing overlay-based DFL 
are \cite{marfoq2020throughput,Huang24MobiHoc} in different network environments. Specifically, \cite{marfoq2020throughput} considered an inter-datacenter network as the underlay where the paths between agents only share links at the first and the last hops, and \cite{Huang24MobiHoc} is our previous work that considered a general multi-hop underlay with arbitrary link sharing. 
This work is closest to \cite{Huang24MobiHoc} with \emph{two important differences}: (i) this work focuses on edge networks where propagation delays are negligible compared to transmission delays, which greatly simplifies the communication optimization (see Section~\ref{subsubsec:Improved Routing Solution}), and (ii) \cite{Huang24MobiHoc} only gave a heuristic algorithm for topology design but this work develops an algorithm with performance guarantee (see Section~\ref{subsec:Link Activation Optimization}). 
We note that a seemingly related work \cite{hua2022efficient} is agnostic to the actual communications in an underlay network, thus not really addressing the overlay setting. 
\looseness=-1

\subsection{Summary of Contributions}

We study the joint optimization of the communication scheme and the communication-controlling hyperparameter for overlay-based DFL atop bandwidth-limited edge networks, with the following contributions:\looseness=-1
\begin{enumerate}
\item  We decompose the overall problem into a subproblem of communication optimization within the overlay network and another subproblem of mixing matrix design for DFL. Using unique characteristics of edge networks, we show that equal bandwidth sharing is optimal under a given overlay routing solution, based on which we simplify the overlay routing problem from a nonlinear optimization to a linear optimization that can minimize the time per iteration under a given mixing matrix design. 


\item We tackle the mixing matrix design problem via sparse convex optimization. By carefully designing the objective function and the solution space, we develop a Frank-Wolfe-type algorithm with guaranteed performance. We also introduce additional steps to further optimize the weights and the search space. 

\item We evaluate the proposed solution in comparison with benchmarks based on parameters from a real edge network. Our results show that: (i) our design can substantially reduce the total training time compared to the baseline without sacrificing the quality of the trained model, and (ii) our proposed algorithm matches the training performance of the state-of-the-art solution with a much lower complexity. 
\end{enumerate}

\textbf{Roadmap.} Section~\ref{sec:Background and Formulation} describes our problem, Section~\ref{sec:Proposed Solution} presents our solution and analysis, Section~\ref{sec:Performance Evaluation} presents our performance evaluation, and Section~\ref{sec:Conclusion} concludes the paper. 
\if\thisismainpaper1
\textbf{All the proofs can be found in Appendix of \cite{Sun25arXiv}.}
\else
\textbf{All the proofs can be found in Appendix~\ref{appendix:Proofs}.} 
\fi


\section{Problem Formulation}\label{sec:Background and Formulation}

\subsection{Notations}\label{subsec:Notations}

Let $\bm{a}\in \mathbb{R}^m$ denote a vector and $\bm{A}\in \mathbb{R}^{m\times m}$ a matrix. We use $\|\bm{a}\|$ to denote the $\ell$-2 norm, and $\|\bm{A}\|$ to denote the spectral norm. We use $\diag(\bm{a})$ to denote a diagonal matrix with the entries in $\bm{a}$ on the main diagonal. We use $\umax(\bm{A})$ and $\vmax(\bm{A})$ to denote the left/right singular vector of $\bm{A}$ corresponding to its largest singular value.  \looseness=-1

\subsection{Network Model}\label{subsec:Network Model}

Consider a network of $m$ learning agents connected through a \emph{base topology} $G=(V, E)$ ($|V|=m$), that forms an overlay on top of a communication underlay $\uG=(\uV, \uE)$ (with $V\subseteq \uV$). For the simplicity of presentation, we assume the overlay to be fully connected, i.e., $E$ contains the links between each pair of nodes in $V$, which is feasible as long as the underlay is connected, but our solution can be easily adapted for non-fully-connected overlays\footnote{This can be achieved by adding linear constraints that force the mixing matrix entries corresponding to non-existing overlay links to zero.}.  
Each underlay link $\ue\in \uE$ has a finite capacity $C_{\ue}$. 
Each overlay link $e=(i,j)\in E$ is implemented via a routing path $\up_{i,j}$ from the node running agent $i$ to the node running agent $j$ in the underlay. 
Let $l_{i, j}$ denote the propagation delay on $\up_{i, j}$. 
For simplicity, we consider both the overlay and the underlay as \emph{undirected} graphs. This effectively means that each underlay link is assumed to have equal capacity in both directions, and each overlay link is assumed to map to symmetric paths (i.e., $\up_{i,j}=\up_{j,i}$), but our solution can be adapted for directed overlay/underlay links as in \cite{Huang24MobiHoc} to model asymmetric capacities and asymmetric routing. 
We assume that \emph{only the overlay nodes in $V$ are controllable}, and the internal nodes in the underlay (i.e., $\uV\setminus V$) are just communication devices (e.g.,  WiFi access points or base stations) and uncontrollable by the learning task. 

\subsection{Learning Task}\label{subsec:Model of Decentralized Learning}

We consider a DFL task, where
each agent $i\in V$ has a possibly non-convex objective function $F_i(\bm{x})$ that depends on the model parameter vector $\bm{x} \in \mathbb{R}^d$ and the local dataset $\mathcal{D}_i$, and the goal is to find the parameter vector $\bm{x}$ that minimizes the global objective function $F(\bm{x})$ defined as\looseness=-1
\begin{align}\label{eq:F(x)}
    F(\boldsymbol{x}) := \frac{1}{m}\sum_{i=1}^{m} F_{i}(\boldsymbol{x}).
\end{align}
For example, we can model the objective of empirical risk minimization by defining the local objective as $F_i(\bm{x}):= \sum_{\bm{s}\in \mathcal{D}_i}\ell(\bm{x},\bm{s})$, where $\ell(\bm{x},\bm{s})$ is the loss function for sample $\bm{s}$ under model $\bm{x}$, and the corresponding global objective is proportional to the empirical risk over all the samples.  
%
%

Suppose that the task is performed by a standard decentralized training algorithm called D-PSGD~\cite{Lian17NIPS}, where each agent repeatedly updates its own parameter vector by SGD and aggregates it with the parameter vectors of its neighbors. Specifically, let $\bm{x}_i^{(k)}$ ($k\geq 1$) denote the parameter vector at agent $i$ after $k-1$ iterations and $g(\bm{x}_i^{(k)}; \xi_i^{(k)})$ the stochastic gradient computed by agent $i$ in iteration $k$, where $\xi_i^{(k)}$ is the mini-batch. In iteration $k$, agent $i$ updates its parameter vector by\looseness=-1
\begin{align}\label{eq:DecenSGD}
\boldsymbol{x}^{(k+1)}_i = \sum_{j=1}^{m}W_{ij}\boldsymbol{x}^{(k)}_j - \eta g(\boldsymbol{x}^{(k)}_i; \xi^{(k)}_i),
\end{align}
where $\bm{W}=(W_{ij})_{i,j=1}^m$ is the $m\times m$ \emph{mixing matrix} in iteration $k$, and $\eta>0$ is the learning rate. 
The update rule in \eqref{eq:DecenSGD} has the same convergence performance as an alternative rule $\boldsymbol{x}^{(k+1)}_i = \sum_{j=1}^{m}W_{ij}(\boldsymbol{x}^{(k)}_j - \eta g(\boldsymbol{x}^{(k)}_j; \xi^{(k)}_j))$~\cite{Lian17NIPS,MATCHA19}, but \eqref{eq:DecenSGD} allows each agent to {parallelize the parameter exchange and the gradient computation}. Since in bandwidth-limited networks, the communication time dominates the computation time~\cite{Luo20MLsys}, the time per iteration according to \eqref{eq:DecenSGD} is determined by the communication time in support of computing $\sum_{j=1}^{m}W_{ij}\boldsymbol{x}^{(k)}_j$. 



\subsection{Design Parameter}\label{subsec:Design Parameter}

The main parameter we focus on is the mixing matrix $\bm{W}$, which directly controls the communication demands as agent $j$ needs to send its parameter vector to agent $i$ if and only if $W_{ij}\neq 0$. According to \cite{Lian17NIPS}, the mixing matrix should be \emph{symmetric with each row/column summing up to one}\footnote{Originally, the mixing matrix was assumed to be symmetric and \emph{doubly stochastic} with entries constrained to $[0,1]$ \cite{Lian17NIPS}, but we find this requirement unnecessary for our adopted convergence bound, which only requires the mixing matrix to be symmetric with each row/column summing up to one.} in order to ensure convergence for D-PSGD. 
The symmetry implies a one-one correspondence between distinct off-diagonal entries in $\bm{W}$ and the overlay links in $E$, and thus $W_{ij}$ can be interpreted as the \emph{link weight} of the overlay link $(i,j)\in E$. Specifically, the requirement of each row summing to one implies that $W_{ii} = 1-\sum_{j=1}^m W_{ij}$. In the vector form, this implies the following form of the mixing matrix 
\begin{align}
\bm{W} = \bm{I} - \bm{B} \diag(\bm{\alpha})\bm{B}^\top, \label{eq:W}
\end{align}
where $\bm{I}$ is the $m\times m$ identity matrix, $\bm{B}$ is the $|V|\times|E|$ incidence matrix\footnote{This is defined under an arbitrary orientation of each link $e_j\in E$ as $B_{i j}= +1$ if $e_j$ starts from $i$, $-1$ if $e_j$ ends at $i$, and $0$ otherwise.} for the base topology $G$, and $\bm{\alpha}:=(\alpha_{ij})_{(i,j)\in E}$ is the vector of (overlay) link weights. It is easy to see that $W_{ij} = W_{ji} = \alpha_{ij}$. By \eqref{eq:W}, the design of mixing matrix contains two decisions: (i) 
the decision of \emph{which overlay links should be activated} (e.g., having non-zero weights), and (ii) the decision of \emph{how much weight to assign to each activated link}. Agents $i$ and $j$ need to exchange parameter vectors \emph{if and only if} link $(i,j)$ is activated (i.e., $\alpha_{ij}\neq 0$).

\subsection{Design Objective}
\label{subsec:Goal: Network-aware Communication Optimization for Decentralized Learning}

Our goal is to jointly optimize both \emph{the mixing matrix} and \emph{the communication scheme} to serve the demands triggered by the mixing matrix, in order to \emph{minimize the total (wall-clock) time} for the learning task to reach a given level of convergence. The former is an application-layer parameter, and the latter is a network-layer parameter, making our problem a cross-layer optimization. However, as the internal nodes in the underlay are uncontrollable, our solution is limited to actions at only overlay nodes. 

\begin{figure}[t!]
\vspace{-.0em}
\centerline{\mbox{\includegraphics[width=.7\linewidth]{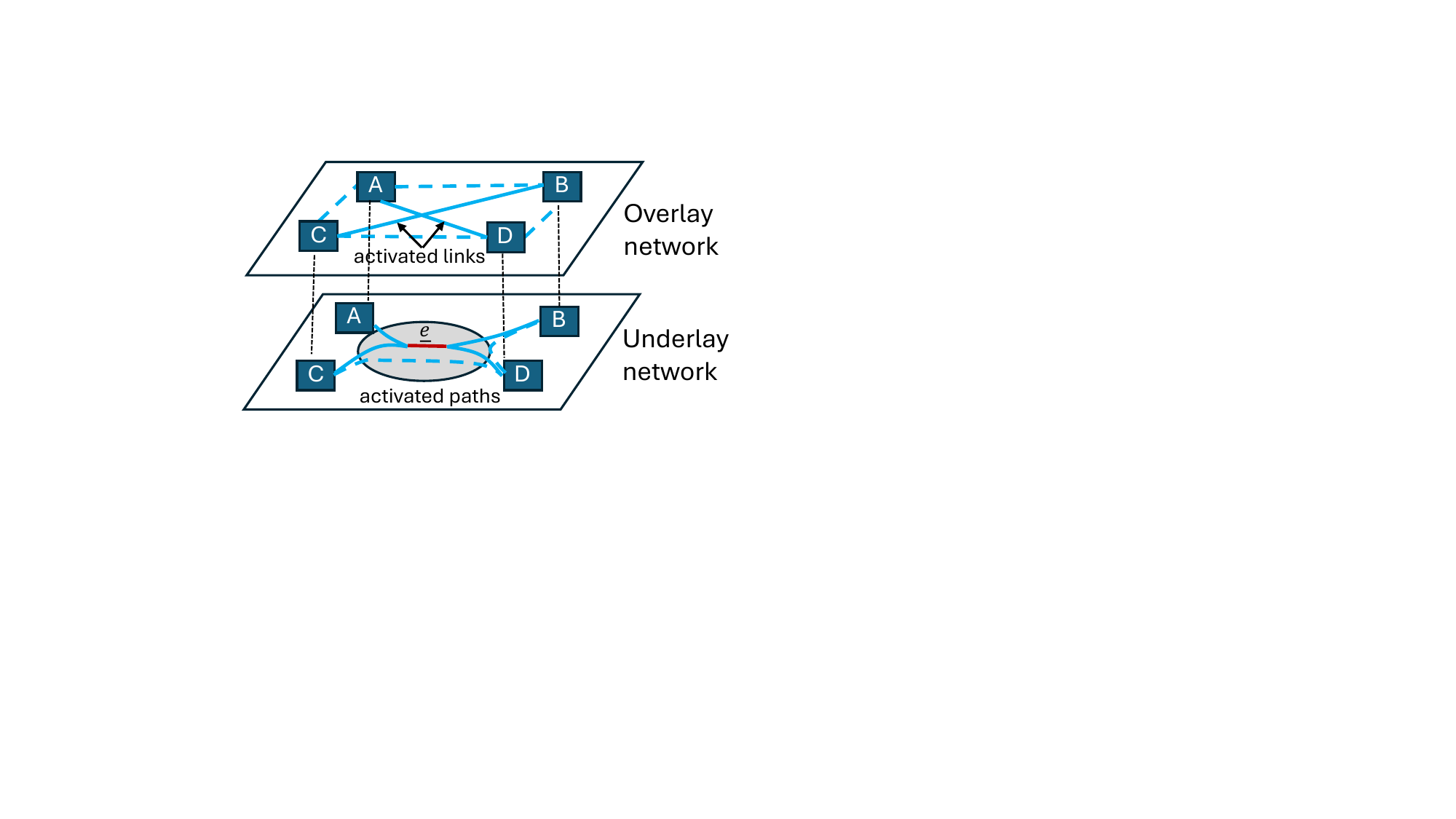}}}
\vspace{-1em}
\caption{\small Example: Overlay routing through B-D-C can accelerate communication by bypassing the shared bottleneck link $\ue$.}
\label{fig:overlay_underlay_routing}
\vspace{-.25em}
\end{figure}

\emph{Remark:} Even if routing inside the underlay is uncontrollable, it is possible to improve performance by controlling only the overlay nodes. For example, in the scenario of Fig.~\ref{fig:overlay_underlay_routing}, if the activated overlay links are $E_a = \{(A,D), (B,C)\}$, directly using the underlay routing paths $\up_{A,D}$ and $\up_{B,C}$ will take a long time to complete the parameter exchange as they share a bottleneck link $\ue$ in the underlay. However, by forwarding the flow between $B$ and $C$ through node $D$, we can bypass the shared bottleneck and allow each flow to complete faster.

\section{Proposed Solution}\label{sec:Proposed Solution}

We will first consider the simpler problem of minimizing the per-iteration time under a given mixing matrix and then address the more complicated problem of optimizing the mixing matrix to balance the per-iteration time and the number of iterations till convergence. 
%
\looseness=-1

\subsection{Overlay Communication Optimization}\label{subsec:Communication Schedule Optimization}

Let $E_a:=\{(i,j)\in E:\: W_{ij}\neq 0\}$ denote the set of activated overlay links under a given mixing matrix. As illustrated by Fig.~\ref{fig:overlay_underlay_routing}, there is room for improvement by optimizing how the overlay nodes collectively serve the communication demands triggered by $E_a$.

\subsubsection{Baseline Solution}

Let $\kappa_i$ denote the size of the parameter vector at agent $i$. Instead of treating the communication demands as a set of unicast flows, with two flows in opposite directions for each activated link $(i,j)\in E_a$, we have shown in \cite{Huang24MobiHoc} that it is more efficient to combine all the flows originating from the same agent $i$ into a single multicast flow disseminating the latest parameters of agent $i$ to other agents it needs to share the parameters with. Thus, the communication demands triggered by a given set of activated links $E_a$ is
\begin{align}\label{eq:demands H}
    H=\{(i,N_{E_a}(i),\kappa_i):\: \forall i\in V \mbox{ with } N_{E_a}(i)\neq \emptyset\},
\end{align} 
where $N_{E_a}(i):=\{j\in V:\: (i,j)\in E_a\}$ is the set of activated neighbors of agent $i$, and each $h=(s_h,T_h,\kappa_h)\in H$ represents a multicast flow with source $s_h$, destinations $T_h$, and data size $\kappa_h$. 

To minimize the total training time, the communication scheme should minimize the time for completing all the demands in $H$ within the control of the overlay. Under the network model in Section~\ref{subsec:Network Model}, this has been achieved in \cite{Huang24MobiHoc} by solving a \emph{mixed integer convex programming (MICP)} problem, summarized below for completeness. 
Let the routing decision be denoted by $z^h_{ij}\in \{0,1\}$ that indicates whether overlay link $(i,j)$ is traversed\footnote{In all the routing-related variables (i.e., $z^h_{ij}, r^{h,k}_{ij}, f^h_{ij}$), the corresponding overlay link $(i,j)$ should be interpreted as a directed link as the flow traversal is directional. Similarly, the link capacity constraint \eqref{min-time:link capacity} should be imposed for each direction of each underlay link.} by the multicast flow $h$ and $r^{h,k}_{ij}\in \{0,1\}$ that indicates whether $(i,j)$ is traversed by the flow from $s_h$ to $k\in T_h$. 
Let the rate decision be denoted by $d_h\geq 0$ that denotes the rate of flow $h$ and $f^h_{ij}\geq 0$ that denotes the rate of flow $h$ on overlay link $(i,j)$. Define constant $b^{h,k}_i$ as $1$ if $i=s_h$, $-1$ if $i=k$, and $0$ otherwise. 
The time in completing all the multicast flows in $H$ \eqref{eq:demands H} can be minimized through the following optimization: 
\begin{subequations}\label{eq:min-time}
\begin{align}
   \min_{\bm{z},\bm{r},\bm{d},\bm{f}} &\quad \tau\\
   \mbox{s.t.} &\quad  \tau \geq {\kappa_h\over d_h}+ \sum_{(i,j)\in E}l_{i, j} r_{ij}^{h,k} ,~~\forall h\in H, k\in T_h,
   \label{min-time:time}\\
&   \sum_{(i,j)\in E: \ue\in \up_{i,j}}\sum_{h\in H} f^h_{ij}\leq C_{\ue},~~\forall \ue\in \uE,   \label{min-time:link capacity}\\   
   & \hspace{-2em}\sum_{j\in V}r^{h,k}_{ij} = \sum_{j\in V}r^{h,k}_{ji} + b^{h,k}_i, ~~ \forall h\in H, k\in T_h, i\in V, \label{min-time:flow conservation}\\
   & r^{h,k}_{ij}\leq z^h_{ij}, ~~ \forall h\in H, k\in T_h, (i,j)\in E, \label{min-time:tree} \\
   &\hspace{-1em} d_h-M(1-z^h_{ij})\leq f^h_{ij}\leq d_h,~\forall h\in H, (i,j)\in E, \label{min-time:big-M}\\
   &f_{ij}^h \le M z_{ij}^h,~\forall h\in H, (i,j)\in E,\label{noFlow:big-M}\\ 
   & \hspace{-0em} r^{h,k}_{ij}, z^h_{ij}\in \{0,1\},~d_h\in [0,M],~f^h_{ij}\geq 0,\nonumber\\
   & \hspace{8em}\forall h\in H, k\in T_h, (i,j)\in E, \label{min-time:bound}
\end{align}
\end{subequations}
where $M$ is an upper bound on $d_h$ ($\forall h\in H$). 
Constraint \eqref{min-time:time} ensures $\tau$ as an upper bound on the completion time; \eqref{min-time:link capacity} enforces the capacity at each underlay link; \eqref{min-time:flow conservation}--\eqref{min-time:tree} are the \emph{Steiner arborescence} constraints \cite{Goemans93Networks} that guarantee the set of overlay links with $z^h_{ij}=1$ will form a directed Steiner tree from $s_h$ to each $k\in T_h$; \eqref{min-time:big-M}--\eqref{noFlow:big-M} ensure that $f^h_{ij}=d_h z^h_{ij}$, which allows the capacity constraint to be written as a linear inequality \eqref{min-time:link capacity}. 
The optimal solution $(\bm{z}^*, \bm{r}^*, \bm{d}^*, \bm{f}^*)$ to \eqref{eq:min-time} thus provides an overlay communication scheme that minimizes the communication time under a given set of activated links.

\emph{Remark:} 
The optimization \eqref{eq:min-time} is a MICP problem with $O(|V|^2 |E|)$ variables  and $O(|\uE| + |V|^2(|V|+|E|))$ constraints, which is challenging to solve for large networks. 

\subsubsection{Improved Solution}\label{subsubsec:Improved Routing Solution} 

In the application scenario of DFL over an edge network, the underlay spans a relatively small area, making the propagation delay $l_{i,j}$ negligible compared to the transmission delay. Moreover, since all the agents are training the same model, the sizes of local parameter vectors will be identical (in absence of compression\footnote{Even if compression is used, we can still set $\kappa_i\equiv \kappa$ in the communication optimization, where $\kappa$ denotes the maximum compressed model size, to obtain a guaranteed per-iteration time.}), i.e., $\kappa_i\equiv \kappa$ ($\forall i\in V$). 
In this scenario, we can express the objective as a closed-form function of the routing variables as follows.

\begin{lemma}\label{lem:equal bandwidth allocation}
Define a \emph{unicast flow activated by $z^h_{ij}=1$} as a flow in the underlay carrying the content of $h\in H$ from agent $i$ to agent $j$. 
Given a routing solution $\bm{z}$, define 
\begin{align}
t_{\ue}:= \sum_{(i,j): \ue\in \up_{ij}}\sum_{h\in H}z^h_{ij}
\end{align}
as the number of {activated unicast flows} 
traversing an underlay link $\ue\in \uE$. 
If $l_{i,j}=0$ $\forall (i,j)\in E$ and $\kappa_h\equiv \kappa$ $\forall h\in H$, then the optimal objective value of \eqref{eq:min-time} under $\bm{z}$ is 
\begin{align}\label{eq:tau - special case, per-link}
\tau = \max_{\ue\in \uE} {\kappa t_{\ue}\over  C_{\ue} },
\end{align}
achieved by \emph{equally sharing the bandwidth} at every underlay link among the activated unicast flows traversing it. 
\end{lemma}

Lemma~\ref{lem:equal bandwidth allocation} implies that for $l_{i,j}=0$ $\forall (i,j)\in E$ and $\kappa_h\equiv \kappa$ $\forall h\in H$, \eqref{eq:min-time} is reduced to the optimization of overlay routing $\bm{z}$, after which the flow rates can be easily determined as $d_h\equiv \min_{\ue\in \uE} C_{\ue}/t_{\ue}$ $\forall h\in H$. 
The reduced optimization has a much simpler form as follows:
\begin{subequations}\label{eq:min-time simple}
\begin{align}
   \min_{\bm{z},\bm{r}} &\quad \tau\\
   \mbox{s.t.} &\quad  \tau \geq {\kappa\over C_{\ue}} \sum_{(i,j): \ue\in \up_{ij}}\sum_{h\in H}z^h_{ij},~~\forall \ue\in \uE,   \label{simple min-time:time}\\ 
   & \mbox{\eqref{min-time:flow conservation}--\eqref{min-time:tree}} \\   
   & \hspace{-0em} r^{h,k}_{ij}, z^h_{ij}\in \{0,1\},~~\forall h\in H, k\in T_h, (i,j)\in E, \label{simple min-time:bound}
\end{align}
\end{subequations}
Although \eqref{eq:min-time simple} has the same order of variables/constraints as \eqref{eq:min-time}, it has a qualitative difference that all the constraints are linear, making \eqref{eq:min-time simple} a \emph{mixed integer linear programming (MILP)} problem that is much easier to solve numerically than \eqref{eq:min-time}. 

\emph{Remark:} Besides simplifying the computation, Lemma~\ref{lem:equal bandwidth allocation} also implies that if each activated unicast flow is implemented as a TCP flow (using the same congestion control algorithm), then the minimum completion time will be automatically achieved as long as the routing is optimal.

\subsubsection{Handling Uncooperative Underlay}

When learning over a third-party-managed network, the overlay cannot directly solve \eqref{eq:min-time simple}, because the constraint \eqref{simple min-time:time} requires the knowledge of the internal state of the underlay (routing paths and link capacities). In absence of such knowledge, we can leverage a result from \cite{Huang23MobiHoc} to convert this constraint into an equivalent form that can be consistently estimated by the overlay. To this end, we introduce the following notion from \cite{Huang23MobiHoc}.   

\begin{definition}[\cite{Huang23MobiHoc}]\label{def: category}
A \textbf{category of underlay links $\Gamma_F$} for a given set of overlay links $F$ ($F\subseteq E$) is the set of underlay links traversed \emph{by and only by} the underlay routing paths for the overlay links in $F$, i.e, 
\begin{align}\label{eq:category}
\Gamma_F\coloneqq \Big(\bigcap_{(i,j)\in {F}}\up_{i,j}\Big) \setminus \Big(\bigcup_{(i,j)\in E\setminus F}\up_{i,j}\Big). 
\end{align}
\end{definition}

The key observation is that since all the underlay links in the same category are traversed by the same set of overlay links, they must carry the same traffic load from the overlay. Therefore, we can reduce the completion time formula in Lemma~\ref{lem:equal bandwidth allocation} to a formula based on category-level information. 

\begin{lemma}\label{lem:equal bandwidth allocation - category}
Given a routing solution $\bm{z}$, define
\begin{align}
t_F := \sum_{(i,j)\in F}\sum_{h\in H}z^h_{ij} 
\end{align}
as the number of activated unicast flows traversing the links in category $\Gamma_F$. If $l_{i,j}=0$ $\forall (i,j)\in E$ and $\kappa_h\equiv \kappa$ $\forall h\in H$, then the optimal objective value of \eqref{eq:min-time} under $\bm{z}$ is
\begin{align}\label{eq:tau - special case, per-category}
\tau = \max_{F\in \mathcal{F}} {\kappa t_F\over  C_F},
\end{align}
achieved by \emph{equally sharing the bandwidth} at every underlay link among the activated unicast flows traversing it, where $\mathcal{F}:=\{F\subseteq E: \Gamma_F\neq \emptyset\}$ and $C_F\coloneqq \min_{\ue \in \Gamma_F}C_{\ue}$.
\end{lemma}

Plugging the result of Lemma~\ref{lem:equal bandwidth allocation - category} into \eqref{simple min-time:time} yields a MILP similar to \eqref{eq:min-time simple}, except that \eqref{simple min-time:time} becomes
\begin{align}\label{eq:simple min-time, category}
\tau\geq {\kappa\over C_F} \sum_{(i,j)\in F}\sum_{h\in H}z^h_{ij}, ~~\forall F\in\mathcal{F}.
\end{align}
The benefit of the new formulation is that instead of requiring detailed internal knowledge about the underlay (i.e., $(\up_{i,j})_{(i,j)\in E}$ and $(C_{\ue})_{\ue\in \uE}$) as in \eqref{simple min-time:time}, the new constraint \eqref{eq:simple min-time, category} only requires the knowledge of the \emph{nonempty categories $\mathcal{F}$} and the corresponding \emph{bottleneck capacity in each category $(C_F)_{F\in\mathcal{F}}$}, both of which can be estimated consistently by the overlay using algorithms from \cite{Huang23MobiHoc}. 
 Given the inferred parameters $\widehat{\mathcal{F}}$ and $(\widehat{C}_F)_{F\in \widehat{\mathcal{F}}}$, we can simply plug them in place of $\mathcal{F}$ and $(C_F)_{F\in\mathcal{F}}$ to compute an overlay routing solution.  

\emph{Remark:} While solving the MILP to optimality can incur super-polynomial complexity, we note that this optimization only needs to be solved once at the beginning of the learning task (by the task orchestrator~\cite{Huang24MobiHoc}), and thus such computation overhead is usually tolerable.

\subsection{Mixing Matrix Design}\label{subsec:Mixing Matrix Design}

As explained in Section~\ref{subsec:Design Parameter}, the mixing matrix design contains two decisions: (i) the design of the links to activate, i.e., $E_a$, and (ii) the design of the weight for each activated link, i.e., $(\alpha_{ij})_{(i,j)\in E_a}$. Since the latter is already solved in \cite{Huang24MobiHoc}, we will just summarize the result for completeness and then focus on the former problem.

\subsubsection{Link Weight Design}\label{subsubsec:Conditional Link Weight Optimization}

The foundation of our design is a state-of-the-art convergence bound for D-PSGD under the following assumptions:
\begin{enumerate}[(1)]
    \item Each local objective function $F_i(\boldsymbol{x})$ is $l$-Lipschitz smooth, i.e.,    
    $\|\nabla F_i(\bm{x})-\nabla F_i(\bm{x}')\|\leq l\|\bm{x}-\bm{x}'\|,\: \forall i\in V$.
\item There exist constants $M_1, \widehat{\sigma}$ such that    
    $\hspace{-0em} {1\over m}\sum_{i\in V}\E[\|g(\bm{x}_i;\xi_i)-\nabla F_i(\bm{x}_i) \|^2] \leq \widehat{\sigma}^2 + {M_1\over m}\sum_{i\in V}\|\nabla F(\bm{x}_i)\|^2$, 
    $\forall \bm{x}_1,\ldots,\bm{x}_m \in \mathbb{R}^d$.    
\item There exist constants $M_2, \widehat{\zeta}$ such that 
    ${1\over m}\sum_{i\in V}\|\nabla F_i(\bm{x})\|^2\leq \widehat{\zeta}^2 + M_2\|\nabla F(\bm{x}) \|^2, \forall \bm{x} \in \mathbb{R}^d$.    
\end{enumerate}
Let $\bm{J}:={1\over m}\bm{1} \bm{1}^\top$ denote an ideal $m\times m$ mixing matrix with all entries being ${1\over m}$.

\begin{theorem}\cite[Theorem~2]{Koloskova20ICML} \label{thm:new convergence bound}
Under assumptions (1)--(3), if the mixing matrix $\bm{W}$, which is symmetric with each row/column summing to one, satisfies that $\rho:=\|\bm{W}-\bm{J}\|<1$, then D-PSGD can achieve $ \frac{1}{K} \sum_{k=1}^K \mathbbm{E}[\|\nabla F(\boldsymbol{\overline{x}}^k)\|^2]\leq \epsilon$ for any given $\epsilon>0$ ($\overline{\bm{x}}^{(k)}:={1\over m}\sum_{i=1}^m \bm{x}^{(k)}_i$) when the number of iterations reaches
\begin{align}
K(\rho)&:= l(F(\overline{\bm{x}}^{(1)})-F_{\inf}) \nonumber\\
&\hspace{-2.75em}   \cdot \hspace{-.15em}O\hspace{-.25em}\left(\hspace{-.25em}{\widehat{\sigma}^2\over m\epsilon^2} \hspace{-.15em}+\hspace{-.15em}{\widehat{\zeta}\sqrt{M_1+1}+\widehat{\sigma}\sqrt{ 1-\rho^2}\over (1-\rho^2) \epsilon^{3/2}} \hspace{-.15em}+\hspace{-.15em} {\sqrt{(M_2+1)(M_1+1)}\over (1-\rho^2)\epsilon} \hspace{-.05em}\right), \label{eq:new bound on K}
\end{align}
where $\overline{\bm{x}}^{(1)}$ is the initial parameter vector, and $F_{\inf}$ is a lower bound on $F(\cdot)$. 
\end{theorem}

\emph{Remark:} 
The original version of \cite[Theorem~2]{Koloskova20ICML} covers more general cases where the mixing matrices can be random and time-varying. However, for the tractability of design, we will focus on the case of a single deterministic mixing matrix, in which case the bound in \cite[Theorem~2]{Koloskova20ICML} is reduced to \eqref{eq:new bound on K} as shown in \cite{Huang24MobiHoc}. 
While there exist other convergence bounds for D-PSGD such as \cite{MATCHA19,vogels2022beyond,neglia2020decentralized,le2023refined}, we choose Theorem~\ref{thm:new convergence bound} as the theoretical foundation of our design due to the generality of its assumptions as explained in \cite{Koloskova20ICML}. 


Based on Theorem~\ref{thm:new convergence bound}, the parameter $\rho$ captures the impact of the mixing matrix on the convergence rate: the smaller $\rho$, the smaller $K(\rho)$. Given the activated links $E_a$, the optimal weights can then be computed by minimizing $\rho$ subject to having zero-weight for nonactivated links as follows:
\begin{subequations}\label{eq:min rho wo cost}
\begin{align}
& \min_{\bm{\alpha}}\:  \rho \label{wo cost:obj}\\
\mbox{s.t. }
& - \rho\bm{I} \preceq \bm{I} - \bm{B} \diag(\bm{\alpha}) \bm{B}^\top - \bm{J} \preceq  \rho\bm{I}, \label{wo cost:matrix}\\
& \alpha_{ij} = 0,~~\forall (i,j)\not\in E_a, \label{wo cost:activation}
\end{align}
\end{subequations}
which is a semi-definite programming (SDP) problem that can be solved in polynomial time. 

\subsubsection{Link Activation Design}\label{subsec:Link Activation Optimization}

The hardest part of mixing matrix design is the design of which links to activate. The difficulty of this problem arises from the fact that the mixing matrix affects both the time per iteration and the number of iterations required to reach convergence. To minimize the total training time, we should ideally solve 
\begin{align}\label{eq:overall obj}
\min_{\bm{W}} \tau(\bm{W}) \cdot K(\rho(\bm{W})),
\end{align}
where $\tau(\bm{W})$ denotes the time per iteration according to \eqref{eq:min-time simple} (with \eqref{simple min-time:time} replaced by \eqref{eq:simple min-time, category}) for the demands triggered by the mixing matrix $\bm{W}$, and $K(\rho(\bm{W}))$ for $\rho(\bm{W}):= \|\bm{W}-\bm{J}\|$ denotes the number of iterations till convergence according to Theorem~\ref{thm:new convergence bound}. 
This optimization not only has a large solution space but also has a non-closed-form objective function.  \looseness=-1

\emph{Sparse Convex Optimization.} 
Our basic idea is to build the set of activated links gradually by converting \eqref{eq:overall obj} into a sparse convex optimization problem. Intuitively, we can bound the per-iteration time $\tau(\bm{W})$ by bounding the number of activated links. Thus, if we can bound the convergence parameter $\rho(\bm{W})$ and thus the number of iterations $K(\rho(\bm{W}))$ while keeping the number of activated links bounded, then we will be able to bound the total training time according to \eqref{eq:overall obj}. 

To this end, we will leverage the Frank-Wolfe method~\cite{Jaggi13ICML}, which is an iterative way of minimizing a convex objective function through gradient-based linearization. This method is particularly useful when the solution space is the convex hull of a large set of basic solutions called ``atoms'', in which case the solution after $T$ iterations will be ``sparse'' in the sense that it is the convex combination of at most $T$ distinct atoms. The key in applying this method is to construct a suitable set of atoms such that their convex hull contains a good solution and each atom is sufficiently sparse. 

In our case, the atoms should be valid mixing matrices that each activates a small number of links. A natural choice with this property is the set of \emph{swapping matrices $\mathcal{S}:=\{\bm{S}^{(i,j)}: (i,j)\in E\}$}, where each $\bm{S}^{(i,j)}$ is an $m\times m$ matrix that equals the identity matrix except that $S^{(i,j)}_{ii}=S^{(i,j)}_{jj}=0$ and $S^{(i,j)}_{ij} = S^{(i,j)}_{ji} = 1$. Using $\bm{S}^{(i,j)}$ as the mixing matrix has the effect of swapping the parameter vectors at $i$ and $j$ by only activating link $(i,j)$. Meanwhile, we can show that the swapping matrices together with the identity matrix can express all the mixing matrices through linear combinations. 

\begin{lemma}\label{lem:linear combination of swapping matrices}
Any mixing matrix $\bm{W}$ can be written as
\begin{align}
\bm{W} = \left(1-\sum_{(i,j)\in E}\alpha_{ij}\right)\bm{I} + \sum_{(i,j)\in E}\alpha_{ij}\bm{S}^{(i,j)},
\end{align}
where $\alpha_{ij}=W_{ij}$ $\forall (i,j)\in E$. 
\end{lemma}

Lemma~\ref{lem:linear combination of swapping matrices} suggests that we can design the mixing matrix by solving the following constrained convex optimization:
\begin{align}\label{eq:Frank-Wolfe}
\min_{\bm{W}\in \conv(\mathcal{S}^+)} \|\bm{W}-\bm{J}\| =: \rho(\bm{W}),
\end{align}
where $\mathcal{S}^+:= \mathcal{S}\cup \{\bm{I}\}$ and $\conv(\cdot)$ denotes the convex hull. 

\begin{algorithm}[tb]
\small
\SetKwInOut{Input}{input}\SetKwInOut{Output}{output}
\Input{Objective function $\rho(\bm{W})$, solution space $\conv(\mathcal{S}^+)$, \#iterations $T$. }
\Output{Designed mixing matrix $\bm{W}^{(T)}$.}
initialize $\bm{W}^{(0)}\leftarrow \bm{I}$ (the identity matrix)\;
\For{$k=0,\ldots,T-1$}
{$\bm{S}^{(k+1)}\leftarrow \argmin_{\bm{W}\in \conv(\mathcal{S}^+)}<\bm{W}, \nabla\rho(\bm{W}^{(k)})>$\; \label{FMMD:select S}
$\bm{W}^{(k+1)}\leftarrow {k\over k+2}\bm{W}^{(k)}+{2\over k+2}\bm{S}^{(k+1)}$\; \label{FMMD:update W}
}
\caption{Frank-Wolfe Mixing Matrix Design (FMMD)}
\vspace{-.0em}
\label{Alg:FMMD}
\end{algorithm}
\normalsize

\emph{Frank-Wolfe-type Algorithm.} 
Applying the Frank-Wolfe method~\cite{Jaggi13ICML} to \eqref{eq:Frank-Wolfe} yields a mixing matrix design algorithm referred to as \emph{Frank-Wolfe Mixing Matrix Design (FMMD)}, as shown in Alg.~\ref{Alg:FMMD}. It iteratively selects the atom with the minimum inner product with the gradient of the objective function (line~\ref{FMMD:select S}) and incorporates the selected atom into the solution through a convex combination (line~\ref{FMMD:update W}). For the problem in \eqref{eq:Frank-Wolfe}, the gradient is given by
\begin{align}
\nabla\rho(\bm{W}^{(k)}) = \umax(\bm{W}^{(k)}-\bm{J}) \vmax^\top(\bm{W}^{(k)}-\bm{J}),
\end{align}
where $\umax(\bm{A})$ and $\vmax(\bm{A})$ denote the left/right singular vector of a matrix $\bm{A}$ corresponding to its largest singular value. Since the objective of line~\ref{FMMD:select S} is linear in $\bm{W}$, the $\argmin$ must be achieved at an atom in $\mathcal{S}^+$. Thus, line~\ref{FMMD:select S} can be implemented by selecting $\bm{S}^{(k+1)}$ as 
\begin{align}
\hspace{-.5em}\argmin_{\bm{S}\in \mathcal{S}^+} \hspace{-.25em}<\hspace{-.25em}\bm{S}, \umax(\bm{W}^{(k)}\hspace{-.15em}-\hspace{-.15em}\bm{J}) \vmax^\top(\bm{W}^{(k)}\hspace{-.15em}-\hspace{-.15em}\bm{J})\hspace{-.25em}>\hspace{-.0em}, \label{eq:selecting S}
\end{align}
which can be computed efficiently as $|\mathcal{S}^+|=O(m^2)$. 

\emph{Performance Guarantee.}
The Frank-Wolfe-type updates of FMMD has the property that the solution $\bm{W}^{(T)}$ after $T$ iterations is the convex combination of up to $T$ atoms in $\mathcal{S}^+$. Since each atom in $\mathcal{S}^+$ activates at most one link, this allows us to bound the number of activated links in $\bm{W}^{(T)}$ and hence the per-iteration time $\tau(\bm{W}^{(T)})$. Together with the bounded optimality gap of the Frank-Wolfe method in approximating the optimal objective value, this leads to the following performance guarantee for FMMD. 

\begin{theorem}\label{thm:FMMD}
Suppose that $l_{i,j}=0, \forall (i,j)\in E$ and $\kappa_i\equiv \kappa, \forall i\in V$. Let $\Cmin := \min_{F\in \mathcal{F}} C_F$. If $m>3$ and $T>{16\over 3}m-2$, then the total training time under the mixing matrix designed by FMMD after $T$ iterations is bounded by
\begin{align}
\hspace{-.5em}\tau(\bm{W}^{(T)}) K(\rho(\bm{W}^{(T)})) \leq {\kappa T\over \Cmin}  K\left({m-3\over m}+{16\over T+2}\right), \label{eq:FMMD bound}
\end{align}
where $K(\cdot)$ is defined as in \eqref{eq:new bound on K}. 
Moreover, when $m\gg 1$ and $T = \lceil {32\over 5}m-2\rceil$, the bound \eqref{eq:FMMD bound} is explicitly characterized as
\begin{align}
&O\bigg({\kappa  l(F(\overline{\bm{x}}^{(1)})-F_{\inf})\over \Cmin} \cdot \nonumber\\
&\hspace{3em} \Big({\widehat{\zeta}\sqrt{M_1+1}\over \epsilon^{3/2}}+{\sqrt{(M_2+1)(M_1+1)}\over \epsilon} \Big) m^2 \bigg). \label{eq:FMMD bound - large m}
\end{align}
\end{theorem}

\emph{Remark:} The bound \eqref{eq:FMMD bound - large m} implies that the total training time will grow quadratically with the number of agents $m$. 

\emph{Further Improvements.}
While the direct application of Frank-Wolfe method allows us to provide a theoretical performance guarantee (Theorem~\ref{thm:FMMD}), the performance of FMMD can be further improved due to the following observations:
\begin{enumerate}
\item Although any mixing matrix can be represented as a linear combination of the atoms in $\mathcal{S}^+$ (Lemma~\ref{lem:linear combination of swapping matrices}), the Frank-Wolfe method only optimizes among the convex combinations of $\mathcal{S}^+$, thus leaving room for improvement by further optimizing the link weights. 
\item The Frank-Wolfe method controls \#activated links, but the actual communication time $\tau$ can differ under the same \#activated links, suggesting the need to consider the impact on $\tau$ when selecting which link to activate. 
\end{enumerate}

Accordingly, we introduce the following improvements to Alg.~\ref{Alg:FMMD}. \emph{First}, we can extract the set of links $E_a(\bm{W}^{(T)}):=\{(i,j)\in E: W^{(T)}_{ij}\neq 0\}$ activated by the mixing matrix $\bm{W}^{(T)}$ designed by the Frank-Wolfe method, and then use \eqref{eq:min rho wo cost} to further optimize the non-zero weights therein. \emph{Moreover}, we can prioritize atoms with small impact on the per-iteration time in line~\ref{FMMD:select S}, by modifying the search space of \eqref{eq:selecting S} from the entire $\mathcal{S}^+$ to 
the subset of unselected atoms that once selected will cause the minimum per-iteration time (i.e., $\min \tau(\bm{W}^{(k+1)})$). However, evaluating $\tau(\bm{W})$ requires solving a MILP \eqref{eq:min-time simple}, which is computationally expensive. Instead, we use an easily computable upper bound on $\tau(\bm{W})$, given by the completion time when routing each activated flow to the default path given by the underlay: 
\begin{align}\label{eq:tau no routing}
\overline{\tau}(\bm{W}):= \max_{F\in \mathcal{F}}{\kappa\over C_F}|E_a(\bm{W})\cap F|.
\end{align}
Based on this bound, we modify the search space of \eqref{eq:selecting S} to
\begin{align}\label{eq:limited space for selecting S}
&\hspace{-.75em}\bigg\{\bm{S}\in \mathcal{S}^+\setminus \mathcal{S}(\bm{W}^{(k)}): \overline{\tau}\Big({k\over k+2}\bm{W}^{(k)}+{2\over k+2}\bm{S}\Big) \leq \nonumber\\
&\hspace{.5em}\overline{\tau}\Big({k\over k+2}\bm{W}^{(k)}+{2\over k+2}\bm{S}'\Big), \forall \bm{S}'\in \mathcal{S}^+\setminus \mathcal{S}(\bm{W}^{(k)})\bigg\},
\end{align}
where $\mathcal{S}(\bm{W}^{(k)})$ denotes the set of (already selected) atoms used to construct $\bm{W}^{(k)}$. 

We refer to FMMD with the first improvement as \emph{FMMD with Weight optimization (FMMD-W)}, FMMD with the second improvement as \emph{FMMD with Priority (FMMD-P)}, and FMMD with both improvements as \emph{FMMD with Weight optimization and Priority (FMMD-WP)}.

\section{Performance Evaluation}\label{sec:Performance Evaluation}

We evaluate the proposed algorithms against benchmarks through realistic data-driven simulations. \looseness=0

\subsection{Simulation Setup}

\subsubsection{Dataset and ML Model}

We train a ResNet-50 model with 23,616,394 parameters and a model size of 94.47 MB (under precision FP32) for image classification on the CIFAR-10 dataset, which contains 60,000 color images divided into 10 classes. We use 50,000 images for training and the remaining 10,000 images for testing. The dataset undergoes standard preprocessing, including normalization and one-hot encoding of the labels.
%
We use a mini‐batch size of 64 and an adaptive learning rate that is 0.1 for the first 30 epochs, 
0.05 for the next 30 epochs, 
and 0.01 thereafter. These settings are sufficient for D-PSGD to achieve convergence under all the evaluated designs. 
To check the generalizability of our results, we also train a 4-layer CNN model with 0.2 learning rate, which has 582,026 parameters and a model size of 2.33 MB, for digit recognition on the MNIST dataset, which contains 60,000 training images and 10,000 testing images. This model setup is based on the approach from \cite{McMahan17AISTATS}.
\if\thisismainpaper1
Due to space limitation, we defer the results on MNIST to \cite{Sun25arXiv} as the observations are similar. 
\else
We defer the results on MNIST to Appendix~\ref{appendix:Additional Evaluations} as the observations are similar. 
\fi

\subsubsection{Network Topology}

We simulate the underlay based on the topology and link attributes of Roofnet~\cite{Roofnet}, which is a WiFi-based wireless mesh network with 38 nodes, 219 links, and a data rate of 1 Mbps. We select $10$ lowest-degree nodes as learning agents (i.e., overlay nodes) and uniformly distribute the training data among them. We assume shortest-path routing (based on hop count) in the underlay. The network topology is illustrated in Fig.~\ref{fig:underlay_topo}.

\begin{figure}[t!]
\vspace{-0em}
\centerline{\mbox{\includegraphics[width=.6\linewidth]{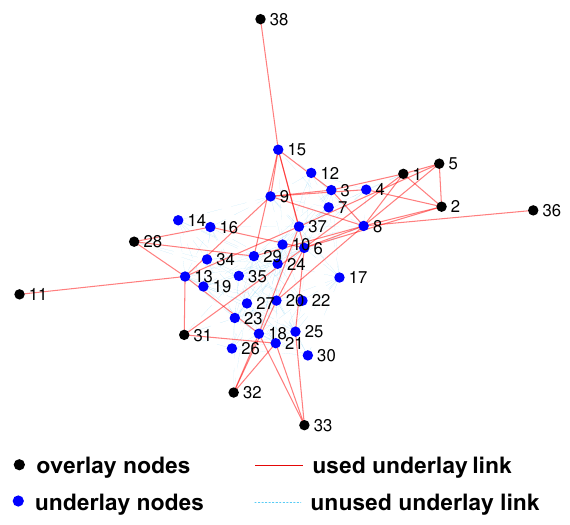}}}
\vspace{-1em}
\caption{\small Network topology and distribution of learning agents. }
\label{fig:underlay_topo}
\vspace{-.5em}
\end{figure}


\subsubsection{Benchmarks}

We compare the proposed algorithm Alg.~\ref{Alg:FMMD} (`FMMD') against the following benchmarks: \begin{itemize}
    \item the baseline of activating all the links (`Clique');
    \item the ring topology (`Ring') commonly used in practice; 
    \item the minimum spanning tree computed by Prim's algorithm (`Prim'),  proposed by \cite{marfoq2020throughput} for DFL over high-bandwidth networks; 
    \item the heuristic based on successive convex approximation (`SCA'), proposed in \cite{Huang24MobiHoc} for DFL over bandwidth-limited networks, which represents the state of the art.    
\end{itemize}

\subsection{Simulation Results}\label{subsec:Simulation Results}

\subsubsection{Comparison of Design Choices}

\begin{figure}[t!]
\begin{minipage}{.495\linewidth}
\centerline{
\includegraphics[width=1\linewidth]{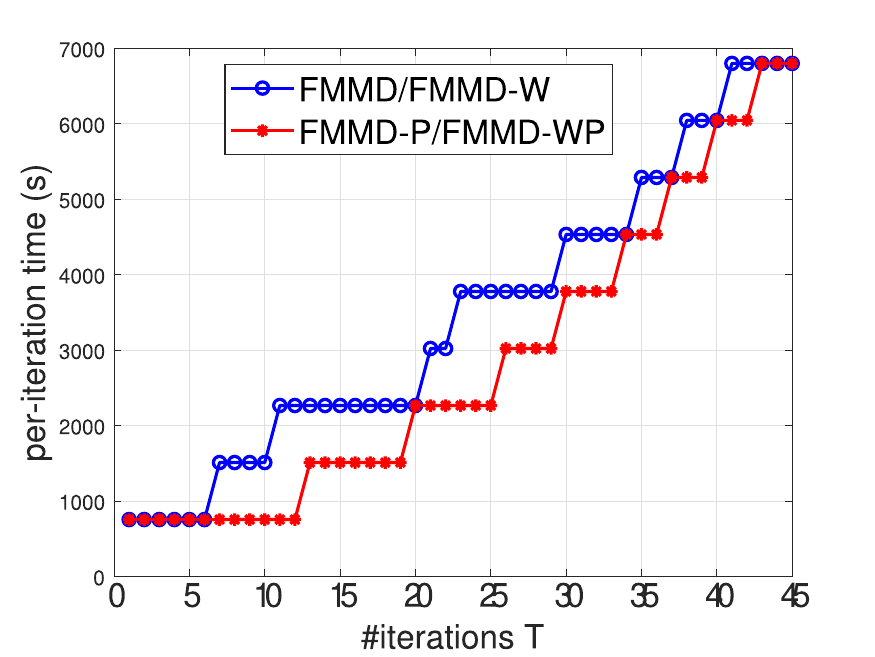}}
\vspace{-.1em}
\end{minipage}
\begin{minipage}{.495\linewidth}
\centerline{
\includegraphics[width=1\linewidth]{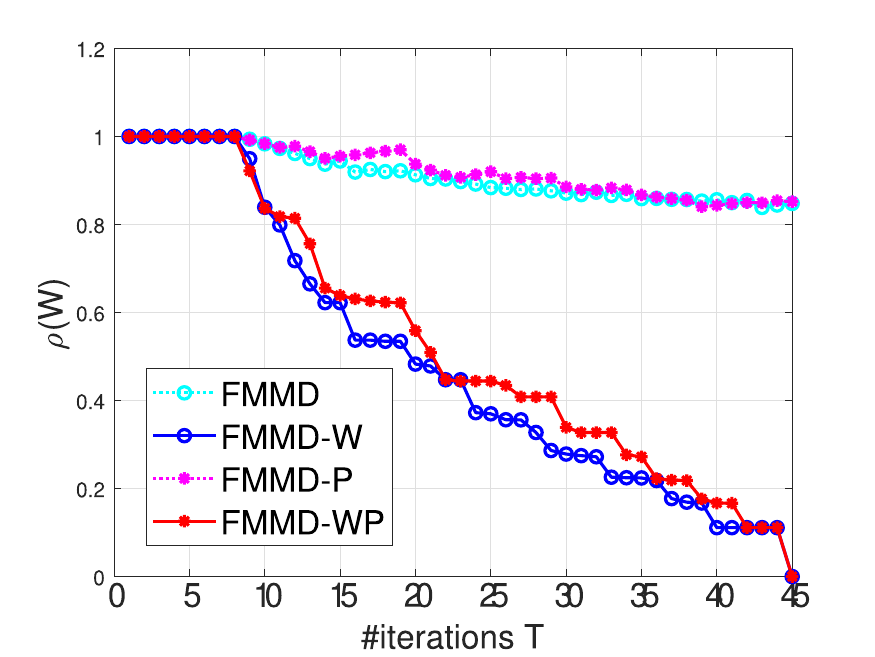}}
\vspace{-.1em}
\end{minipage}
\vspace{-.5em}
\caption{Comparison of FMMD and its variations (note that FMMD and FMMD-W activate the same links and thus have the same per-iteration time, so are FMMD-P and FMMD-WP).  
} \label{fig:FMMD_T}
\vspace{-.5em}
\end{figure}

Fig.~\ref{fig:FMMD_T} compares the vanilla version of FMMD as in Alg.~\ref{Alg:FMMD} (`FMMD') with its variations that incorporate link weight optimization (`FMMD-W'), atom priority (`FMMD-P'), or both (`FMMD-WP'). As these algorithms all aim at designing a communication-efficient mixing matrix $\bm{W}$ with the minimum $\rho(\bm{W})$, we compare them in terms of the per-iteration time $\tau(\bm{W})$ and the convergence rate represented by $\rho(\bm{W})$ (the smaller $\rho$, the faster the convergence). The results show that: (i) all these algorithms achieve smaller $\rho$-values and larger per-iteration times as the number of iterations increases, reflecting the tradeoff between per-iteration communication cost and convergence rate; (ii) further optimizing the weights of activated links by \eqref{eq:min rho wo cost} is necessary for achieving a small $\rho$-value; (iii) while limiting the search space as in \eqref{eq:limited space for selecting S} may cause a slight degradation in the $\rho$-value, it can reduce the per-iteration time by as much as 3$\times$, thus achieving a better tradeoff than using the full search space.  
Since FMMD-WP is the best-performing version, below we will use it to represent the proposed solution, simply referred to as `FMMD'.

\subsubsection{Comparison with Benchmarks}

\begin{figure}[t!]
\begin{minipage}{.495\linewidth}
\centerline{
\includegraphics[width=1\linewidth]{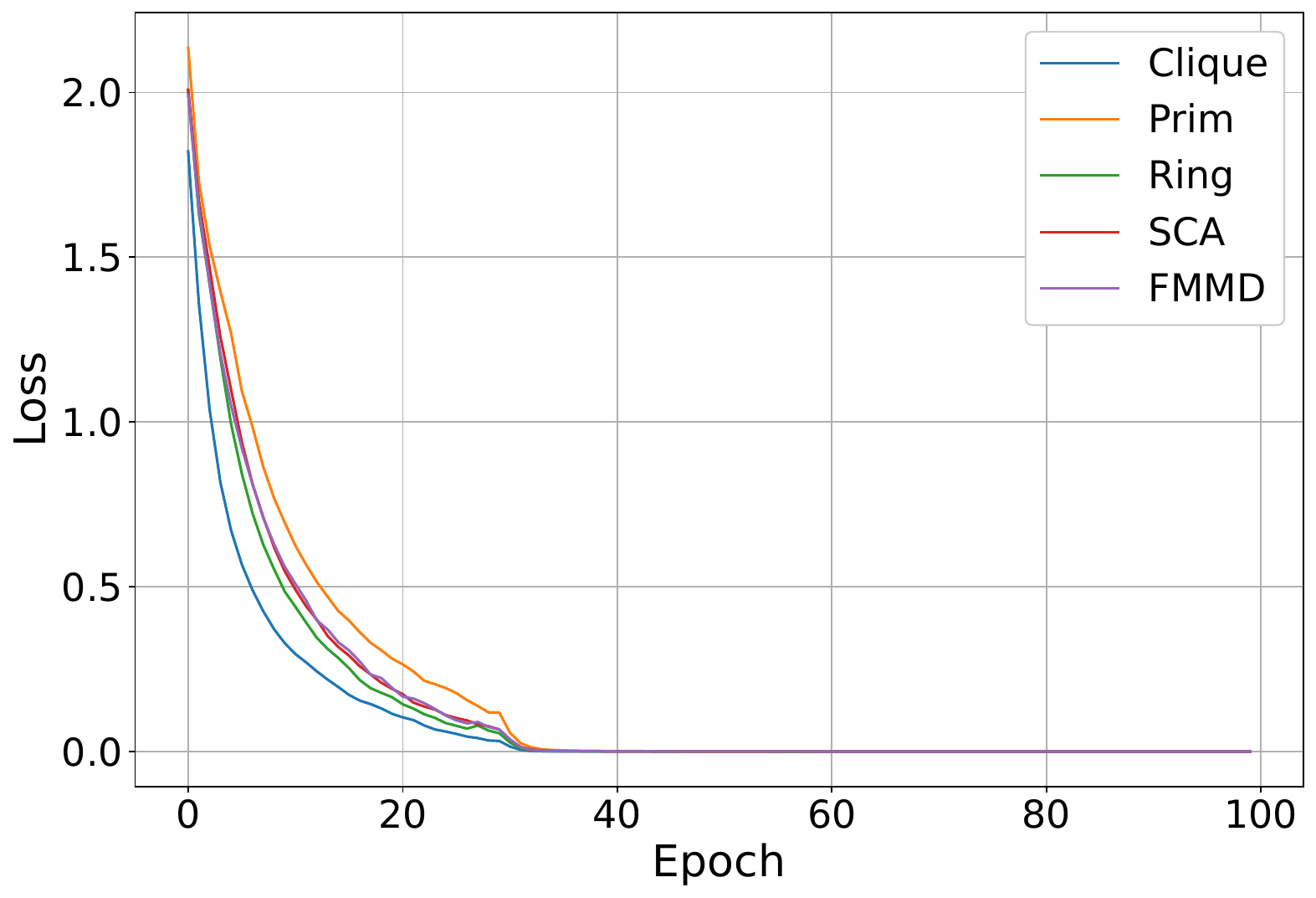}}
\vspace{-.1em}
\end{minipage}
\begin{minipage}{.495\linewidth}
\centerline{
\includegraphics[width=1\linewidth]{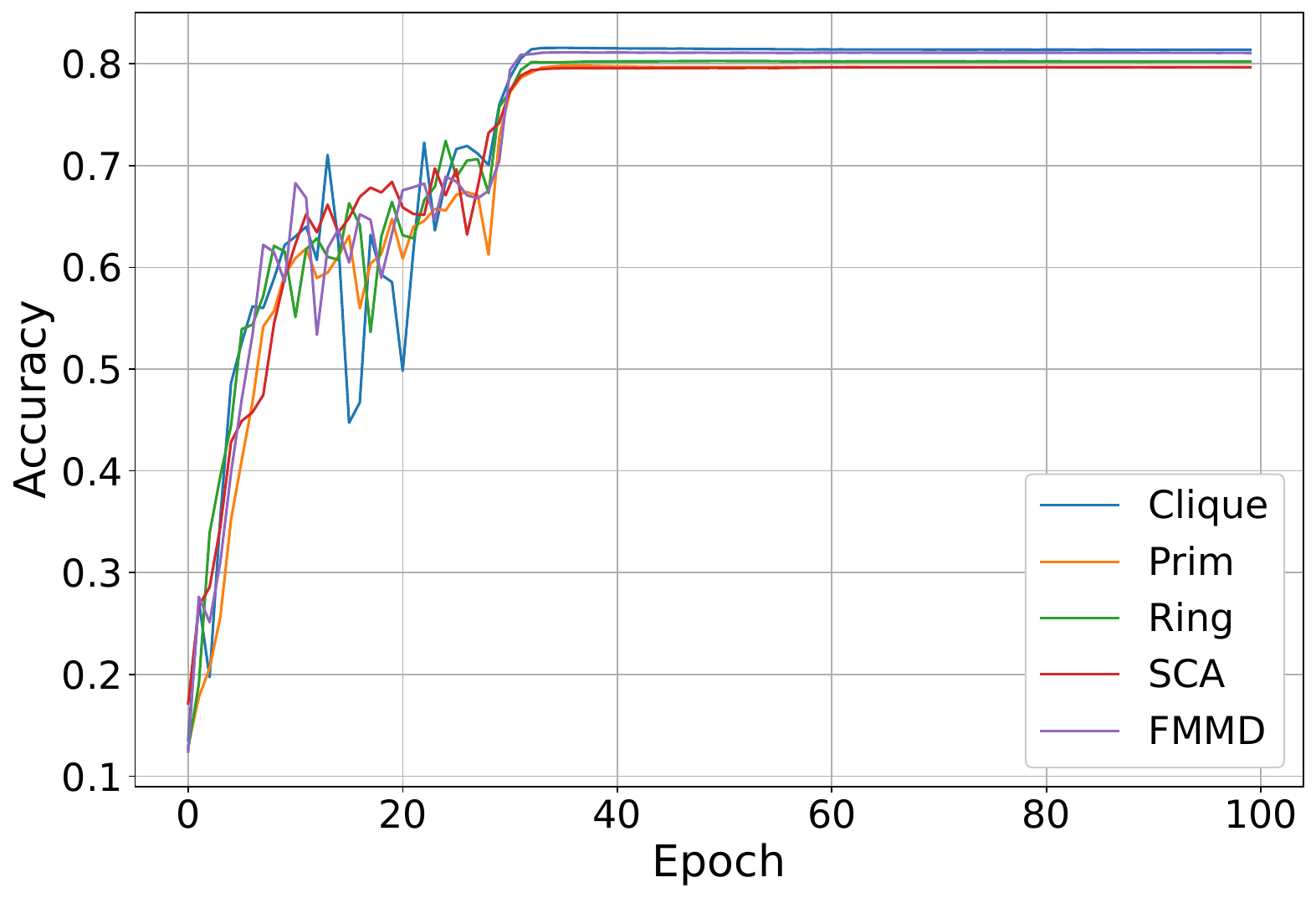}}
\vspace{-.1em}
\end{minipage}
\begin{minipage}{.495\linewidth}
\centerline{
\includegraphics[width=1\linewidth]{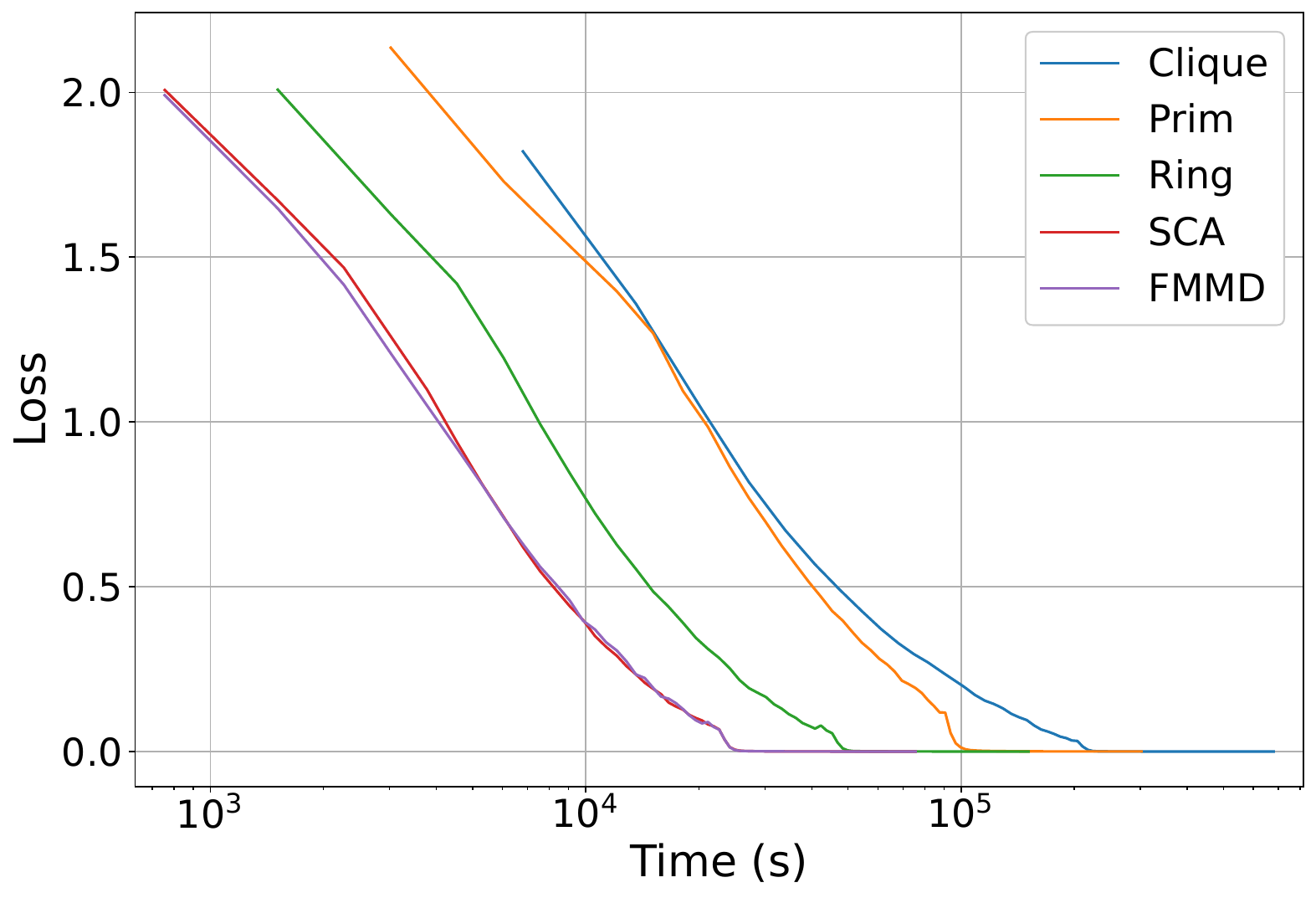}}
\vspace{-.1em}
\end{minipage}
\begin{minipage}{.495\linewidth}
\centerline{
\includegraphics[width=1\linewidth]{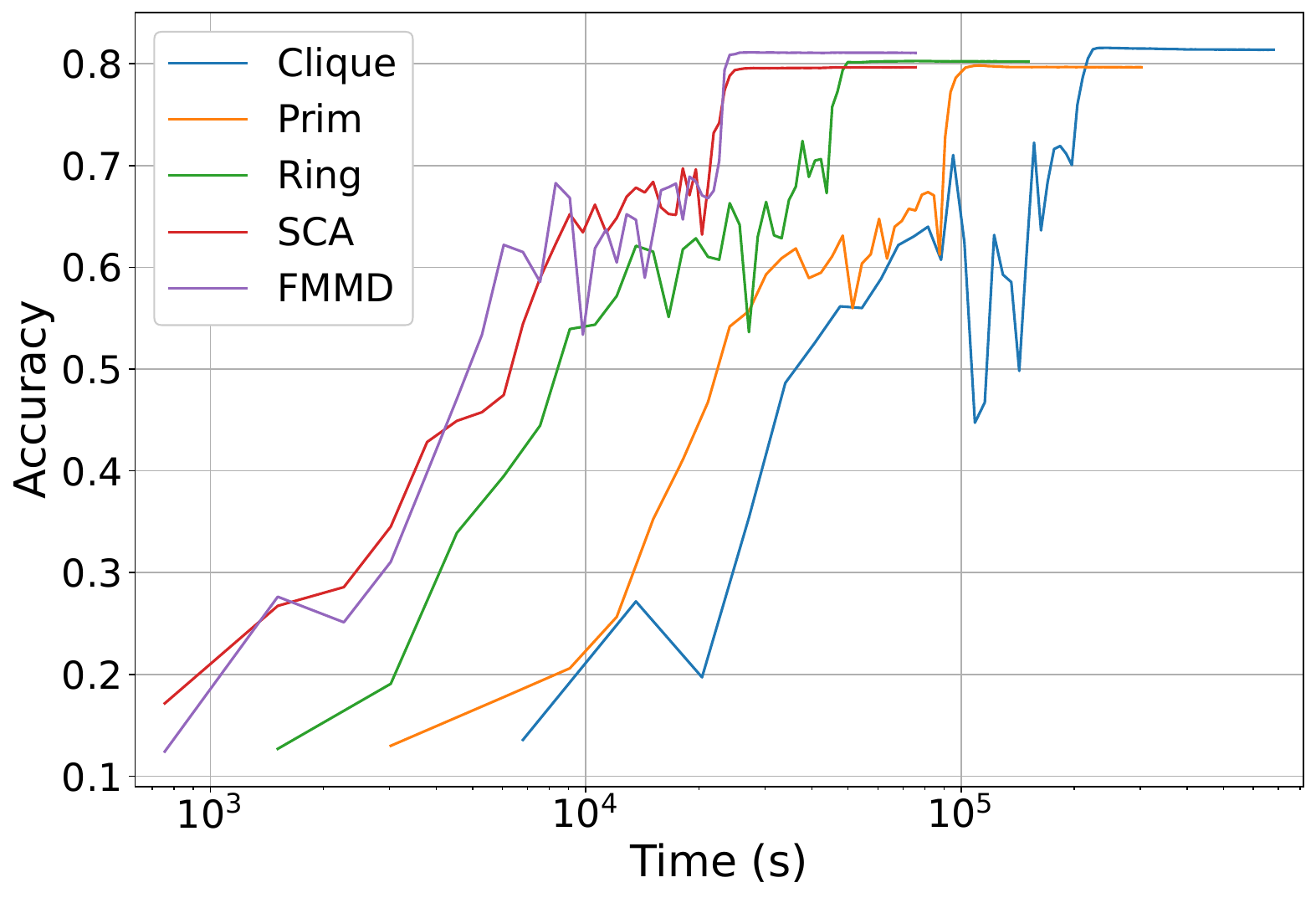}}
\vspace{-.1em}
\end{minipage}
\begin{minipage}{.495\linewidth}
\centerline{
\includegraphics[width=1\linewidth]{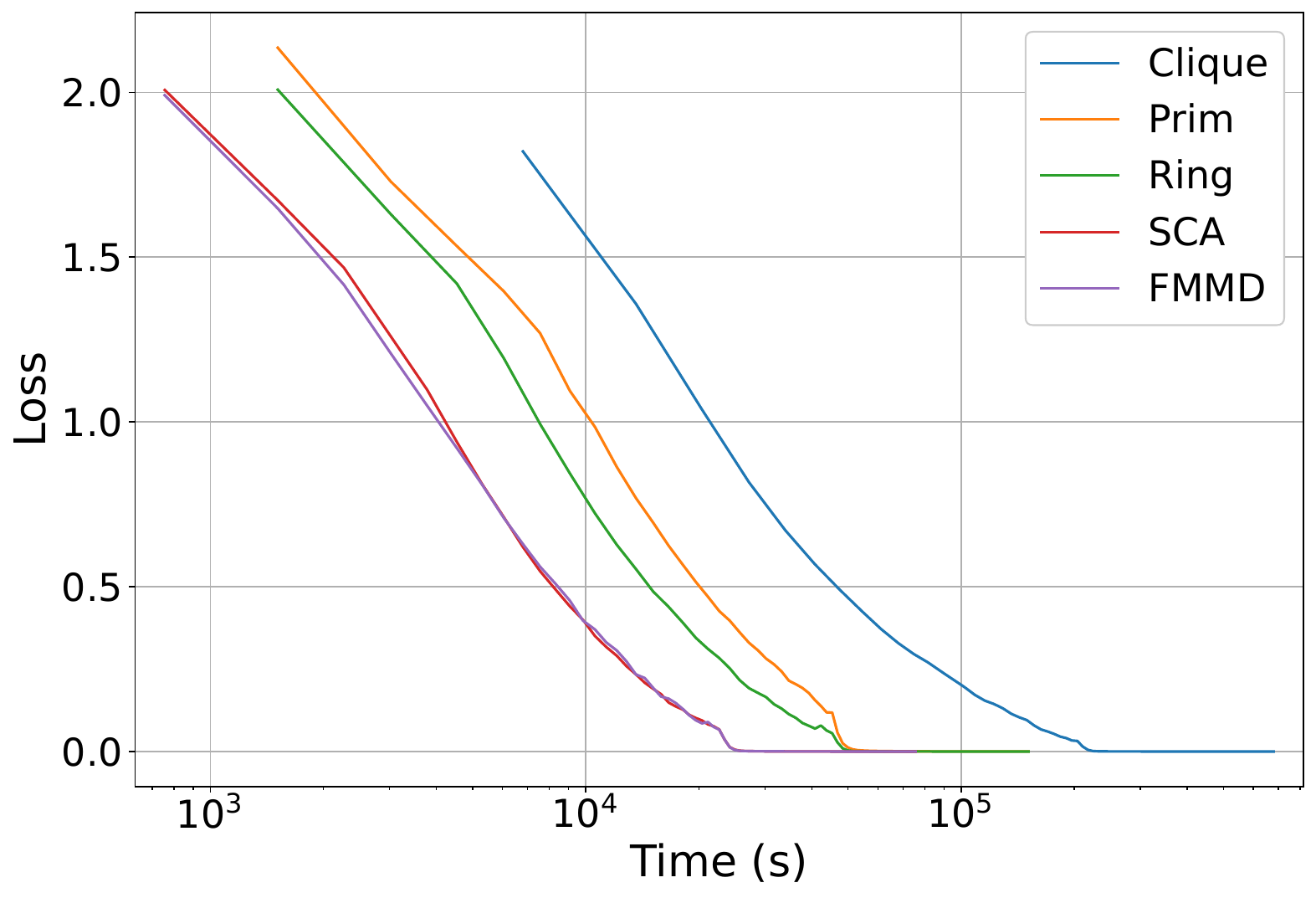}}
\vspace{-.1em}
\end{minipage}
\vspace{-.1em}
\begin{minipage}{.495\linewidth}
\centerline{
\includegraphics[width=1\linewidth]{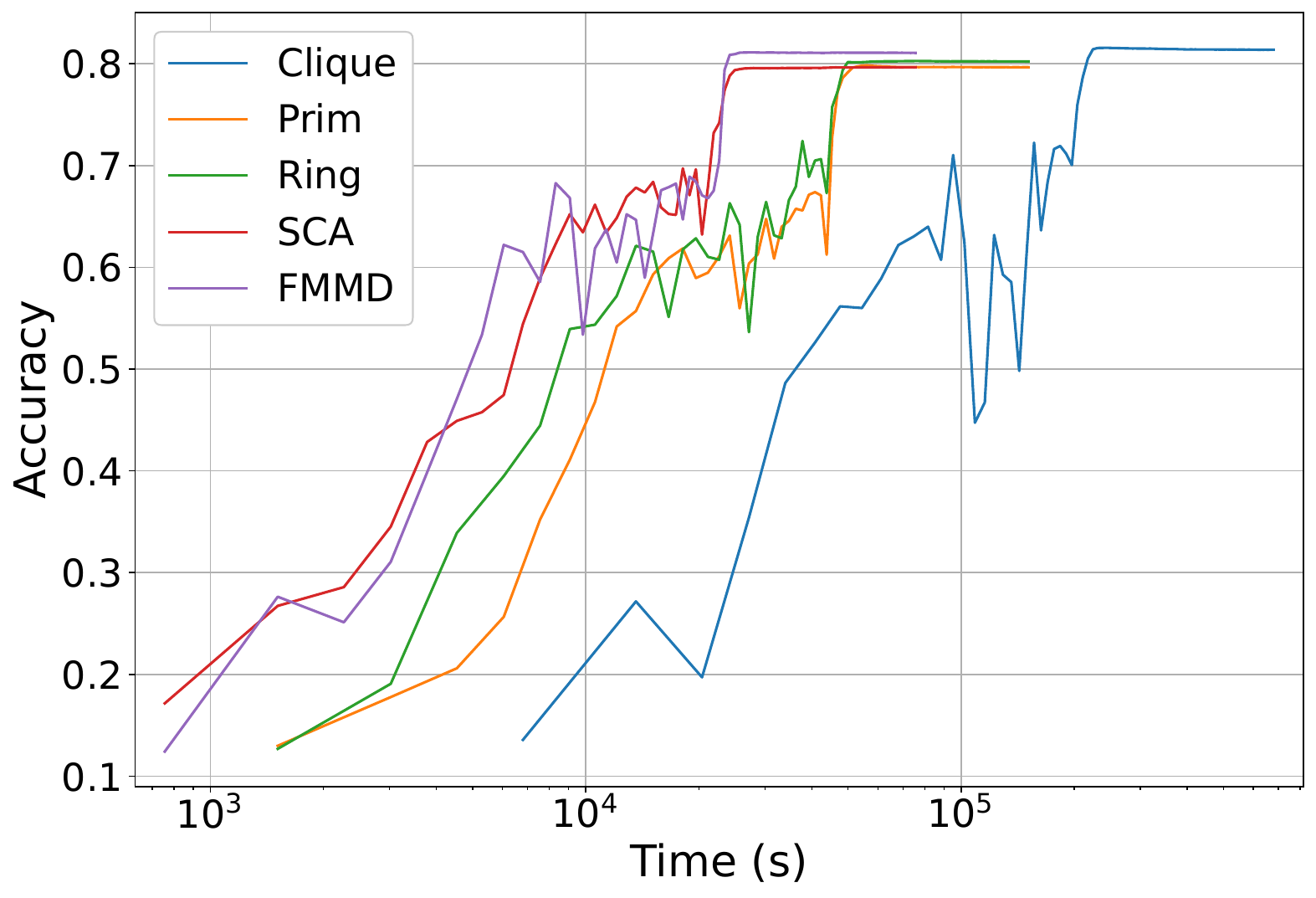}}
\vspace{-.1em}
\end{minipage}
\vspace{-1em}
\caption{Comparison with benchmarks: first row - loss/accuracy over \#epochs, second row - loss/accuracy over $\overline{\tau}$, third row - loss/accuracy over $\tau$.  
} \label{fig:roofnet_no_error}
\vspace{-.0em}
\end{figure}

Fig.~\ref{fig:roofnet_no_error} compares the training performance under various designs in terms of training loss and testing accuracy, where we set \#iterations $T=12$ for `FMMD'. For a fair comparison, we have used \eqref{eq:min rho wo cost} to optimize the link weights under each design. We have also used the same communication scheme for all the designs, either directly communicating over the default routing paths as in \eqref{eq:tau no routing} (second row) or using the optimal overlay routing as in \eqref{eq:min-time simple}, with \eqref{simple min-time:time} replaced by \eqref{eq:simple min-time, category} (third row). The results show that: (i) using sparse topologies rather than the clique can effectively reduce the training time without compromising the performance at convergence, (ii) different designs differ slightly in the convergence rate in terms of epochs, but significantly in the convergence rate in terms of the actual wall-clock time, (iii) the proposed design by `FMMD' matches the best-performing state of the art (`SCA'~\cite{Huang24MobiHoc}), while significantly outperforming the other benchmarks (reducing the training time by {$50\%$} over `Ring' and `Prim' and {$89\%$} over `Clique'), and (iv) overlay routing  can significantly reduce the learning time in some cases (e.g., by $50\%$ for `Prim')\footnote{We note that the exact amount of reduction varies case by case, depending on the underlay topology, routing, and link capacities, and the communication demands in the overlay, but overlay routing always performs no worse than directly using the underlay routing paths.}. We note that among these designs, only `FMMD' and `SCA' consider the internal state of the underlay, which highlights \emph{the importance of network awareness for overlay-based DFL}. 

Meanwhile, there is a computation cost for network awareness as shown in Table~\ref{tab: run_time_algs}, where both `FMMD' and `SCA' are slower than the simplistic designs (`Prim', `Ring', and `Clique') under the routing by \eqref{eq:min-time simple}. Nevertheless, `FMMD' is notably faster than `SCA' thanks to its linearization of the objective function. Moreover, the newly proposed MILP formulation \eqref{eq:min-time simple} is much faster to solve than the previously proposed MICP formulation \eqref{eq:min-time} in \cite{Huang24MobiHoc}. The improved computational efficiency together with the theoretical performance guarantee makes `FMMD' a more desirable solution than `SCA'. \looseness=-1

\begin{table}[]
\small
    \centering
    \begin{tabular}{c|c|c|c|c|c}
         &  SCA & FMMD & Prim & Ring & Clique \\
         \hline
routing by \eqref{eq:min-time} & 2.94  & 2.52 & 3.06 & 4.82 & $>1000$\footnotemark\\       
routing by \eqref{eq:min-time simple} & 1.58  & 0.75 & 0.307 & 0.342 & 0.46
    \end{tabular}
    \vspace{.5em}
    \caption{Running times (s) (including link activation design, link weight design, and overlay routing). 
    }
    \label{tab: run_time_algs}
    \vspace{-1em}
\end{table} 
\footnotetext{When using Gurobi to solve \eqref{eq:min-time}, the solver does not converge in 1,000 s.}

\section{Conclusion}\label{sec:Conclusion}

This work addressed the problem of communication optimization for DFL on top of a bandwidth-limited edge network. Treating the learning agents as overlay nodes, we formulated a joint optimization of both the communication scheme within the overlay and the mixing matrix that controls the communication demands. We showed that the communication scheme design problem can be formulated as a MILP that can be solved without cooperation from the underlay, and the mixing matrix design problem can be formulated as a sparse convex optimization that can be solved by a Frank-Wolfe-type algorithm with guaranteed performance. Our evaluations based on real topology and data validated the ability of the proposed solution to significantly reduce the training time without sacrificing the quality of the trained model. Our overlay-based approach makes our solution easily deployable without requiring the cooperation of the edge network. 


\bibliographystyle{IEEEtran}
\bibliography{references.bib, survey_backup.bib}

\begin{thebibliography}{10}
\providecommand{\url}[1]{#1}
\csname url@samestyle\endcsname
\providecommand{\newblock}{\relax}
\providecommand{\bibinfo}[2]{#2}
\providecommand{\BIBentrySTDinterwordspacing}{\spaceskip=0pt\relax}
\providecommand{\BIBentryALTinterwordstretchfactor}{4}
\providecommand{\BIBentryALTinterwordspacing}{\spaceskip=\fontdimen2\font plus
\BIBentryALTinterwordstretchfactor\fontdimen3\font minus \fontdimen4\font\relax}
\providecommand{\BIBforeignlanguage}[2]{{%
\expandafter\ifx\csname l@#1\endcsname\relax
\typeout{** WARNING: IEEEtran.bst: No hyphenation pattern has been}%
\typeout{** loaded for the language `#1'. Using the pattern for}%
\typeout{** the default language instead.}%
\else
\language=\csname l@#1\endcsname
\fi
#2}}
\providecommand{\BIBdecl}{\relax}
\BIBdecl

\bibitem{Lian17NIPS}
X.~Lian, C.~Zhang, H.~Zhang, C.-J. Hsieh, W.~Zhang, and J.~Liu, ``Can decentralized algorithms outperform centralized algorithms? a case study for decentralized parallel stochastic gradient descent,'' in \emph{Proceedings of the 31st International Conference on Neural Information Processing Systems}, 2017, p. 5336–5346.

\bibitem{McMahan17AISTATS}
H.~McMahan, E.~Moore, D.~Ramage, S.~Hampson, and B.~A. y~Arcas, ``Communication-efficient learning of deep networks from decentralized data,'' in \emph{AISTATS}, 2017.

\bibitem{Kairouz21book}
P.~Kairouz \emph{et~al.}, \emph{Advances and Open Problems in Federated Learning}.\hskip 1em plus 0.5em minus 0.4em\relax Now Foundations and Trends, 2021.

\bibitem{chen2022federated}
X.~Chen, G.~Zhu, Y.~Deng, and Y.~Fang, ``Federated learning over multihop wireless networks with in-network aggregation,'' \emph{IEEE Transactions on Wireless Communications}, vol.~21, no.~6, pp. 4622--4634, 2022.

\bibitem{Compression1}
A.~Koloskova, T.~Lin, S.~U. Stich, and M.~Jagg, ``Decentralized deep learning with arbitrary communication compression,'' in \emph{The International Conference on Learning Representations (ICLR)}, 2020.

\bibitem{Chiu23JSAC}
C.-C. Chiu, X.~Zhang, T.~He, S.~Wang, and A.~Swami, ``Laplacian matrix sampling for communication- efficient decentralized learning,'' \emph{IEEE Journal on Selected Areas in Communications}, vol.~41, no.~4, pp. 887--901, 2023.

\bibitem{Singh21JSAIT}
N.~Singh, D.~Data, J.~George, and S.~Diggavi, ``{SQuARM-SGD}: Communication-efficient momentum {SGD} for decentralized optimization,'' \emph{IEEE Journal on Selected Areas in Information Theory}, vol.~2, no.~3, pp. 954--969, 2021.

\bibitem{MATCHA19}
J.~Wang, A.~K. Sahu, Z.~Yang, G.~Joshi, and S.~Kar, ``{MATCHA}: Speeding up decentralized sgd via matching decomposition sampling,'' in \emph{2019 Sixth Indian Control Conference}.\hskip 1em plus 0.5em minus 0.4em\relax IEEE, 2019, pp. 299--300.

\bibitem{hua2022efficient}
Y.~Hua, K.~Miller, A.~L. Bertozzi, C.~Qian, and B.~Wang, ``Efficient and reliable overlay networks for decentralized federated learning,'' \emph{SIAM Journal on Applied Mathematics}, vol.~82, no.~4, pp. 1558--1586, 2022.

\bibitem{le2023refined}
B.~Le~Bars, A.~Bellet, M.~Tommasi, E.~Lavoie, and A.-M. Kermarrec, ``Refined convergence and topology learning for decentralized {SGD} with heterogeneous data,'' in \emph{International Conference on Artificial Intelligence and Statistics}.\hskip 1em plus 0.5em minus 0.4em\relax PMLR, 2023, pp. 1672--1702.

\bibitem{xing2021federated}
H.~Xing, O.~Simeone, and S.~Bi, ``Federated learning over wireless device-to-device networks: Algorithms and convergence analysis,'' \emph{IEEE Journal on Selected Areas in Communications}, vol.~39, no.~12, pp. 3723--3741, 2021.

\bibitem{pinyoanuntapong2022toward}
P.~Pinyoanuntapong, W.~H. Huff, M.~Lee, C.~Chen, and P.~Wang, ``Toward scalable and robust {AIoT} via decentralized federated learning,'' \emph{IEEE Internet of Things Magazine}, vol.~5, no.~1, pp. 30--35, 2022.

\bibitem{pei2023fed}
J.~Pei, W.~Liu, L.~Wang, C.~Liu, A.~K. Bashir, and Y.~Wang, ``Fed-iout: Opportunities and challenges of federated learning in the internet of underwater things,'' \emph{IEEE Internet of Things Magazine}, vol.~6, no.~1, pp. 108--112, 2023.

\bibitem{jia2022dispatching}
Z.~Jia, Z.~Yu, H.~Liao, Z.~Wang, Z.~Zhou, X.~Wang, G.~He, S.~Mumtaz, and M.~Guizani, ``Dispatching and control information freshness-aware federated learning for simplified power iot,'' in \emph{GLOBECOM}.\hskip 1em plus 0.5em minus 0.4em\relax IEEE, 2022, pp. 1097--1102.

\bibitem{marfoq2020throughput}
O.~Marfoq, C.~Xu, G.~Neglia, and R.~Vidal, ``Throughput-optimal topology design for cross-silo federated learning,'' \emph{Advances in Neural Information Processing Systems}, vol.~33, pp. 19\,478--19\,487, 2020.

\bibitem{Huang23MobiHoc}
Y.~Huang and T.~He, ``Overlay routing over an uncooperative underlay,'' in \emph{The 24th International Symposium on Theory, Algorithmic Foundations, and Protocol Design for Mobile Networks and Mobile Computing (MobiHoc'23)}, 2023, pp. 151--160.

\bibitem{Huang24MobiHoc}
Y.~Huang, T.~Sun, and T.~He, ``Overlay-based decentralized federated learning in bandwidth-limited networks,'' in \emph{Proceedings of the Twenty-Fifth International Symposium on Theory, Algorithmic Foundations, and Protocol Design for Mobile Networks and Mobile Computing (MobiHoc)}, 2024.

\bibitem{ICMLhonor}
Y.~Lu and C.~D. Sa, ``Optimal complexity in decentralized training,'' in \emph{International Conference on Machine Learning (ICML)}, 2021.

\bibitem{sysml19}
J.~Wang and G.~Joshi, ``Adaptive communication strategies to achieve the best error-runtime trade-off in local-update {SGD},'' in \emph{Systems for ML}, 2019.

\bibitem{Wang19JSAC}
S.~Wang, T.~Tuor, T.~Salonidis, K.~K. Leung, C.~Makaya, T.~He, and K.~Chan, ``Adaptive federated learning in resource constrained edge computing systems,'' \emph{IEEE Journal on Selected Areas in Communications}, vol.~37, no.~6, pp. 1205--1221, 2019.

\bibitem{Singh20CDC}
N.~Singh, D.~Data, J.~George, and S.~Diggavi, ``{SPARQ-SGD}: Event-triggered and compressed communication in decentralized optimization,'' \emph{IEEE Transactions on Automatic Control}, vol.~68, no.~2, pp. 721--736, 2022.

\bibitem{Nedic18IEEE}
A.~Nedić, A.~Olshevsky, and M.~G. Rabbat, ``Network topology and communication-computation tradeoffs in decentralized optimization,'' \emph{Proceedings of the IEEE}, vol. 106, no.~5, pp. 953--976, 2018.

\bibitem{Neglia19INFOCOM}
G.~Neglia, G.~Calbi, D.~Towsley, and G.~Vardoyan, ``The role of network topology for distributed machine learning,'' in \emph{IEEE INFOCOM}.\hskip 1em plus 0.5em minus 0.4em\relax IEEE, 2019, pp. 2350--2358.

\bibitem{neglia2020decentralized}
G.~Neglia, C.~Xu, D.~Towsley, and G.~Calbi, ``Decentralized gradient methods: does topology matter?'' in \emph{International Conference on Artificial Intelligence and Statistics}.\hskip 1em plus 0.5em minus 0.4em\relax PMLR, 2020, pp. 2348--2358.

\bibitem{jiang2023joint}
Z.~Jiang, Y.~Xu, H.~Xu, L.~Wang, C.~Qiao, and L.~Huang, ``Joint model pruning and topology construction for accelerating decentralized machine learning,'' \emph{IEEE Transactions on Parallel and Distributed Systems}, 2023.

\bibitem{vogels2022beyond}
T.~Vogels, H.~Hendrikx, and M.~Jaggi, ``Beyond spectral gap: The role of the topology in decentralized learning,'' \emph{Advances in Neural Information Processing Systems}, vol.~35, pp. 15\,039--15\,050, 2022.

\bibitem{Wang22Networking}
J.~Wang, B.~Liang, Z.~Zhu, E.~T. Fapi, and H.~Dalal, ``Joint consensus matrix design and resource allocation for decentralized learning,'' in \emph{2022 IFIP Networking Conference (IFIP Networking)}, 2022, pp. 1--9.

\bibitem{Luo20MLsys}
L.~Luo, P.~West, J.~Nelson, A.~Krishnamurthy, and L.~Ceze, ``Plink: Discovering and exploiting locality for accelerated distributed training on the public cloud,'' in \emph{Proceedings of Machine Learning and Systems}, vol.~2, 2020, pp. 82--97.

\bibitem{Goemans93Networks}
M.~X. Goemans and Y.-S. Myung, ``A catalog of steiner tree formulations,'' \emph{Networks}, vol.~23, no.~1, pp. 19--28, 1993.

\bibitem{Koloskova20ICML}
A.~Koloskova, N.~Loizou, S.~Boreiri, M.~Jaggi, and S.~Stich, ``A unified theory of decentralized {SGD} with changing topology and local updates,'' in \emph{ICML}, 2020.

\bibitem{Jaggi13ICML}
M.~Jaggi, ``Revisiting {Frank-Wolfe}: Projection-free sparse convex optimization,'' in \emph{Proceedings of the 30th International Conference on Machine Learning (ICML)}, 2013, pp. 427--435.

\bibitem{Roofnet}
D.~Aguayo, J.~Bicket, S.~Biswas, G.~Judd, and R.~Morris, ``Link-level measurements from an 802.11b mesh network,'' in \emph{SIGCOMM}, 2004.

\end{thebibliography}

\if\thisismainpaper0

\appendix

\subsection{Supporting Proofs}\label{appendix:Proofs}

\begin{proof}[Proof of Lemma~\ref{lem:equal bandwidth allocation}]
The rate of each multicast flow $h\in H$ is determined by the minimum rate of the unicast flows constituting it. Consider the bottleneck underlay link $\ue^* := \argmin_{\ue\in \uE} C_{\ue}/t_{\ue}$. Since there are $t_{\ue^*}$ unicast flows sharing a total bandwidth of $C_{\ue^*}$ at $\ue^*$, the slowest of these flows cannot have a rate higher than $C_{\ue^*}/t_{\ue^*}$. Thus, the multicast flow containing this slowest unicast flow cannot have a rate higher than $C_{\ue^*}/t_{\ue^*}$, which means that the completion time for all the multicast flows is no smaller than \eqref{eq:tau - special case, per-link}. 
Meanwhile, if the bandwidth of every link is shared equally among the activated unicast flows traversing it, then each unicast flow will receive a bandwidth of no less than $C_{\ue^*}/t_{\ue^*}$ at every hop, and thus can achieve a rate of at least $C_{\ue^*}/t_{\ue^*}$. Hence, each multicast flow $h\in H$ can achieve a rate of at least $C_{\ue^*}/t_{\ue^*}$, yielding a completion time of no more than \eqref{eq:tau - special case, per-link}. 
\end{proof}

\begin{proof}[Proof of Lemma~\ref{lem:equal bandwidth allocation - category}]
According to Lemma~\ref{lem:equal bandwidth allocation}, it suffices to prove that $\min_{F\in \mathcal{F}} C_F/t_F = \min_{\ue\in \uE} C_{\ue}/t_{\ue}$. To this end, we first note that by Definition~\ref{def: category}, all the underlay links in the same category must be traversed by the same set of activated unicast flows, i.e., $t_{\ue} = t_F$ $\forall \ue\in \Gamma_F$. By the definition of the category capacity $C_F$, we have 
\begin{align}
\min_{\ue\in \Gamma_F} {C_{\ue}\over t_{\ue}} = \min_{\ue\in \Gamma_F} {C_{\ue}\over t_F} = {C_F\over t_F}.
\end{align}
Thus, we have
\begin{align}
\min_{\ue\in \uE} {C_{\ue}\over t_{\ue}} = \min_{F\in \mathcal{F}} \min_{\ue\in \Gamma_F} {C_{\ue}\over t_{\ue}} = \min_{F\in \mathcal{F}} {C_F\over t_F}. 
\end{align}
\end{proof}

\begin{proof}[Proof of Lemma~\ref{lem:linear combination of swapping matrices}]
Let $\bm{L}^{(i,j)}$ denote the Laplacian matrix for an $m$-node graph with a single undirected link $(i,j)$. By the definition \eqref{eq:W}, any mixing matrix $\bm{W}$
can be written as
\begin{align}
\bm{W} &= \bm{I}-\sum_{(i,j)\in E}\alpha_{ij}\bm{L}^{(i,j)} \nonumber\\
& = \left(1-\sum_{(i,j)\in E}\alpha_{ij} \right)\bm{I} + \sum_{(i,j)\in E}\alpha_{ij}(\bm{I}-\bm{L}^{(i,j)}).
\end{align}
The proof completes by noting that the swapping matrix $\bm{S}^{(i,j)}= \bm{I}-\bm{L}^{(i,j)}$. 
\end{proof}

\begin{proof}[Proof of Theorem~\ref{thm:FMMD}]
First, as each parameter exchange costs at most $\kappa/\Cmin$ in time, and each iteration of FMMD activates at most one more parameter exchange, the per-iteration time after $T$ iterations is bounded by 
\begin{align}\label{eq:FMMD proof - tau bound}
\tau(\bm{W}^{(T)})\leq {\kappa T\over \Cmin}.
\end{align}

Meanwhile, by \cite[Theorem~1]{Jaggi13ICML}, the optimality gap of $\bm{W}^{(T)}$ is bounded by
\begin{align}\label{eq:FMMD proof - Frank-Wolfe}
\rho(\bm{W}^{(T)})-\rho(\bm{W}^*)\leq {2 C_{\rho}\over T+2},
\end{align}
where $\bm{W}^*$ is the optimal solution to \eqref{eq:Frank-Wolfe}, and $C_{\rho}$ is the curvature constant of $\rho(\bm{W})$ on $\conv(\mathcal{S}^+)$. 

Let $\bm{W}^o$ denote the unconstrained minimum point of $\rho(\bm{W})$. Since $\rho(\bm{W})$ is convex, $\bm{W}^*$ is the Euclidean projection of $\bm{W}^o$ to $\conv(\mathcal{S}^+)$. Equivalently, representing each $\bm{W}\in \conv(\mathcal{S}^+)$ by $\sum_{\bm{S}\in \mathcal{S}^+}\beta_{\bm{S}}\bm{S}$ for $\bm{\beta}\in \Delta_{|\mathcal{S}^+|}$ (the $|\mathcal{S}^+|$-dimensional probability simplex), we need to find the projection of $\bm{\beta}^o$ corresponding to $\bm{W}^o$ to $\Delta_{|\mathcal{S}^+|}$. It is easy to see that the minimum value of $\rho(\bm{W}^o)=0$ is achieved by\looseness=0
\begin{align}
\beta^o_{\bm{S}}=\left\{\begin{array}{ll} 
{3-m\over 2} &\mbox{if } \bm{S}=\bm{I},\\
{1\over m} & \mbox{o.w.}
\end{array}\right.
\end{align}
Its Euclidean projection to $\Delta_{|\mathcal{S}^+|}$ equals
\begin{align}
\beta^*_{\bm{S}}=\left\{\begin{array}{ll} 
0 &\mbox{if } \bm{S}=\bm{I},\\
{2\over m(m-1)} & \mbox{o.w.}
\end{array}\right.
\end{align}
Accordingly, $\bm{W}^*=\sum_{\bm{S}\in\mathcal{S}^+}\beta^*_{\bm{S}}\bm{S}$ satisfies
\begin{align}
W^*_{ij} - J_{ij}=\left\{\begin{array}{ll}
{m-3\over m} & \mbox{if }i=j,\\
{3-m\over m(m-1)} & \mbox{o.w.},
\end{array}\right.
\end{align}
and thus $\rho(\bm{W}^*)=\|\bm{W}^*-\bm{J}\| = (m-3)/m$. 

By \cite{Jaggi13ICML}, $C_{\rho}\leq \diam(\conv(\mathcal{S}^+))^2 L$, where $\diam(\conv(\mathcal{S}^+))$ is the diameter of the solution space, and $L$ is the Lipschitz constant of $\nabla\rho(\bm{W})$. We have
\begin{align}
 \diam(\conv(\mathcal{S}^+)) &= \max_{\bm{S}_1,\bm{S}_2\in \mathcal{S}^+}\|\bm{S}_1-\bm{S}_2\| \nonumber\\
 &\leq 2\max_{\bm{S}\in \mathcal{S}^+}\|\bm{S}\| = 2,
\end{align}
as the spectral norm of any swapping matrix is $1$. The Lipschitz constant is bounded by
\begin{align}
L&\leq \max_{\bm{W}\in \conv(\mathcal{S}^+)}\|\bm{W}-\bm{J}\| \nonumber\\
&\leq \max_{\bm{W}\in \conv(\mathcal{S}^+)}\|\bm{W}\| + \|\bm{J}\|\leq 2,
\end{align}
as $\|\bm{J}\|=1$ and by Jensen's inequality $\|\bm{W}\|\leq \sum_{\bm{S}\in\mathcal{S}^+}\beta_{\bm{S}}\|\bm{S}\|= 1$. 
Thus, $C_{\rho}\leq 8$. 
Plugging the value of $\rho(\bm{W}^*)$ and the bound on $C_{\rho}$ into \eqref{eq:FMMD proof - Frank-Wolfe} yields
\begin{align}\label{eq:FMMD proof - rho bound}
\rho(\bm{W}^{(T)})\leq {m-3\over m}+{16\over T+2}.
\end{align}

Combining \eqref{eq:FMMD proof - tau bound}, \eqref{eq:FMMD proof - rho bound}, and the fact that $K(\cdot)$ in \eqref{eq:new bound on K} is an increasing function leads to the bound in \eqref{eq:FMMD bound}. 

In the case of $m\gg 1$ and $T = \lceil {32\over 5}m-2\rceil$, the bound in \eqref{eq:FMMD proof - rho bound} $\approx 1-{1\over 2m}$, which implies that
\begin{align}
{K\left({m-3\over m}+{16\over T+2}\right) \over l(F(\overline{\bm{x}}^{(1)})-F_{\inf})} &= O\bigg({\widehat{\sigma}^2\over m\epsilon^2}+{\widehat{\zeta}\sqrt{M_1+1}+\widehat{\sigma}\sqrt{1\over m}\over \epsilon^{3/2}/m} \nonumber\\
&\hspace{6em} + {\sqrt{(M_2+1)(M_1+1)}\over \epsilon/m} \bigg) \nonumber\\
&\hspace{-4em}= O\left(\Big({\widehat{\zeta}\sqrt{M_1+1}\over \epsilon^{3/2}}+{\sqrt{(M_2+1)(M_1+1)}\over \epsilon} \Big)m \right). \nonumber
\end{align}
Plugging this as well as $T=O(m)$ into \eqref{eq:FMMD bound} yields \eqref{eq:FMMD bound - large m}.
\end{proof}

\subsection{Additional Evaluations}\label{appendix:Additional Evaluations}

To validate the generalizability of our observations in Section~\ref{subsec:Simulation Results}, we repeat the simulations in Fig.~\ref{fig:FMMD_T}--\ref{fig:roofnet_no_error} and Table~\ref{tab: run_time_algs} for training a 4-layer CNN over the MNIST dataset.
Fig.~\ref{fig:FMMD_T_mnist} shows the comparison between the vanilla version of FMMD and its variations,   Fig.~\ref{fig:roofnet_no_error_mnist} compares the training performance of the best-performing version of FMMD with benchmarks, 
and Table~\ref{tab: run_time_algs_MNIST} compares the algorithm running times. All three results suggest the same observations as before, i.e., FMMD-WP is the best-performing version of the FMMD algorithm, which together with the linearized overlay routing \eqref{eq:min-time simple} can match the best-performing state of the art (`SCA'~\cite{Huang24MobiHoc}) in terms of training performance with significantly better computational efficiency.

\begin{figure}[t!]
\begin{minipage}{.495\linewidth}
\centerline{
\includegraphics[width=1\linewidth]{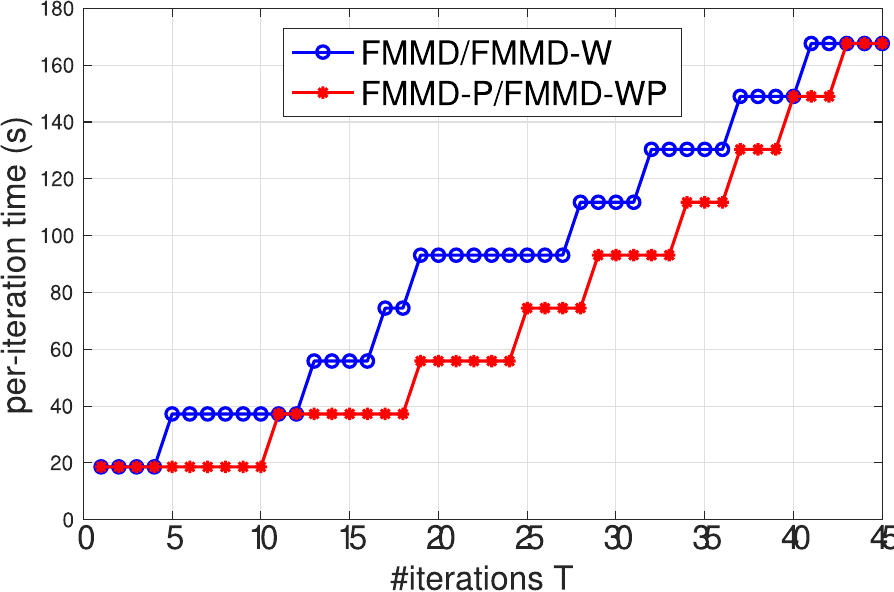}}
\vspace{-.1em}
\end{minipage}
\begin{minipage}{.495\linewidth}
\centerline{
\includegraphics[width=1\linewidth]{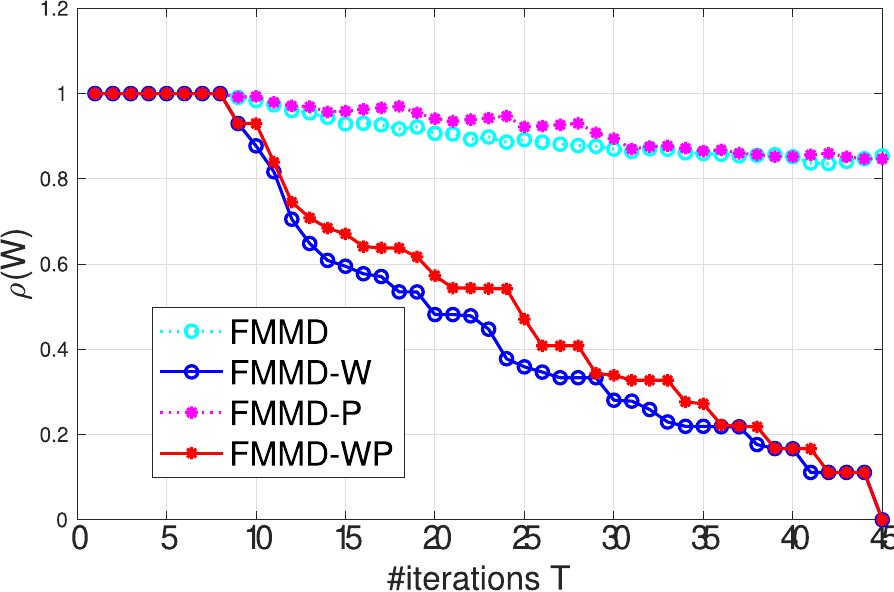}}
\vspace{-.1em}
\end{minipage}
\vspace{-.5em}
\caption{Comparison of FMMD and its variations on MNIST (note that FMMD and FMMD-W activate the same links and thus have the same per-iteration time, so are FMMD-P and FMMD-WP).  
} \label{fig:FMMD_T_mnist}
\vspace{-.5em}
\end{figure}

\begin{figure}[t!]
\begin{minipage}{.495\linewidth}
\centerline{
\includegraphics[width=1\linewidth]{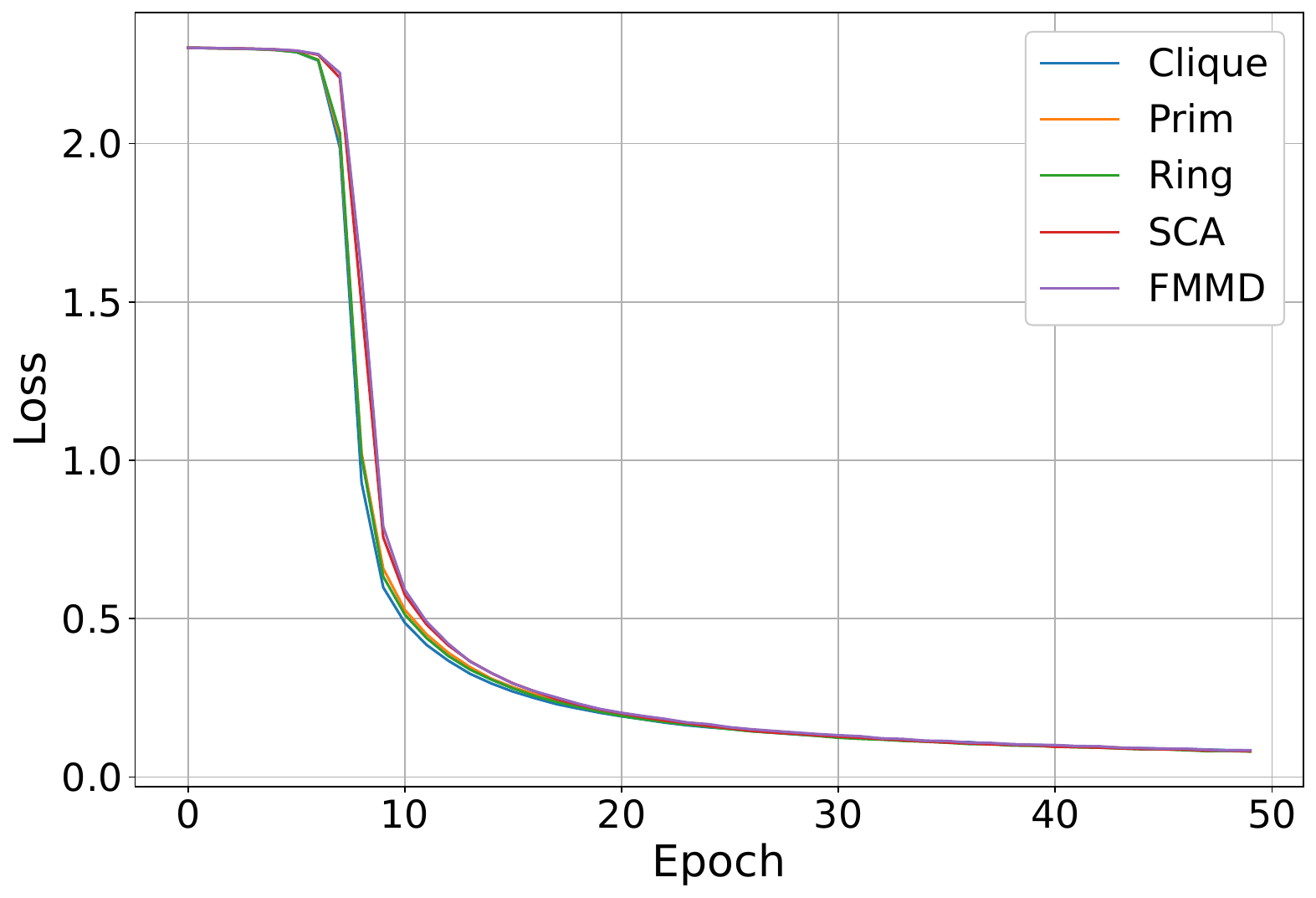}}
\vspace{-.1em}
\end{minipage}
\begin{minipage}{.495\linewidth}
\centerline{
\includegraphics[width=1\linewidth]{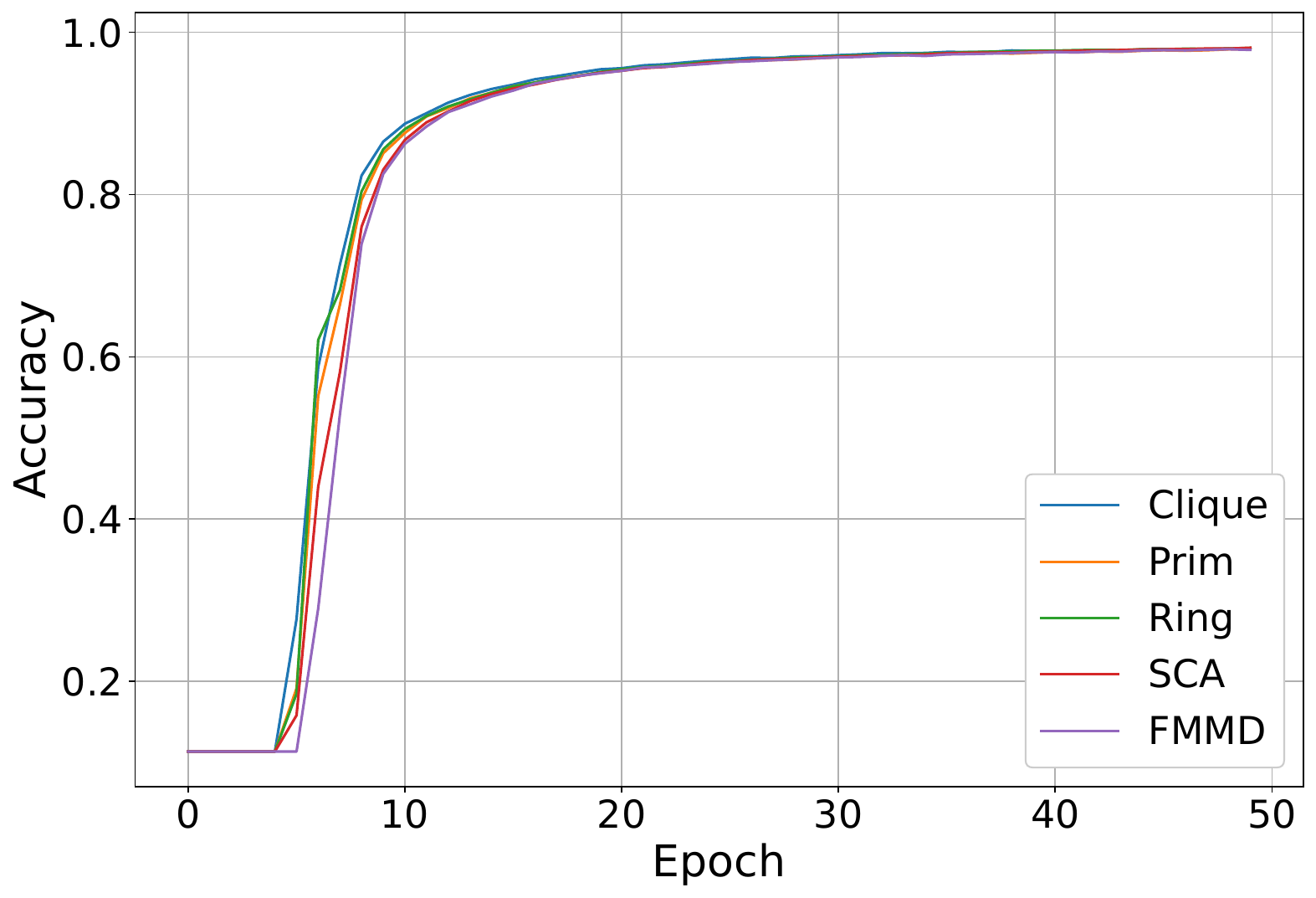}}
\vspace{-.1em}
\end{minipage}
\begin{minipage}{.495\linewidth}
\centerline{
\includegraphics[width=1\linewidth]{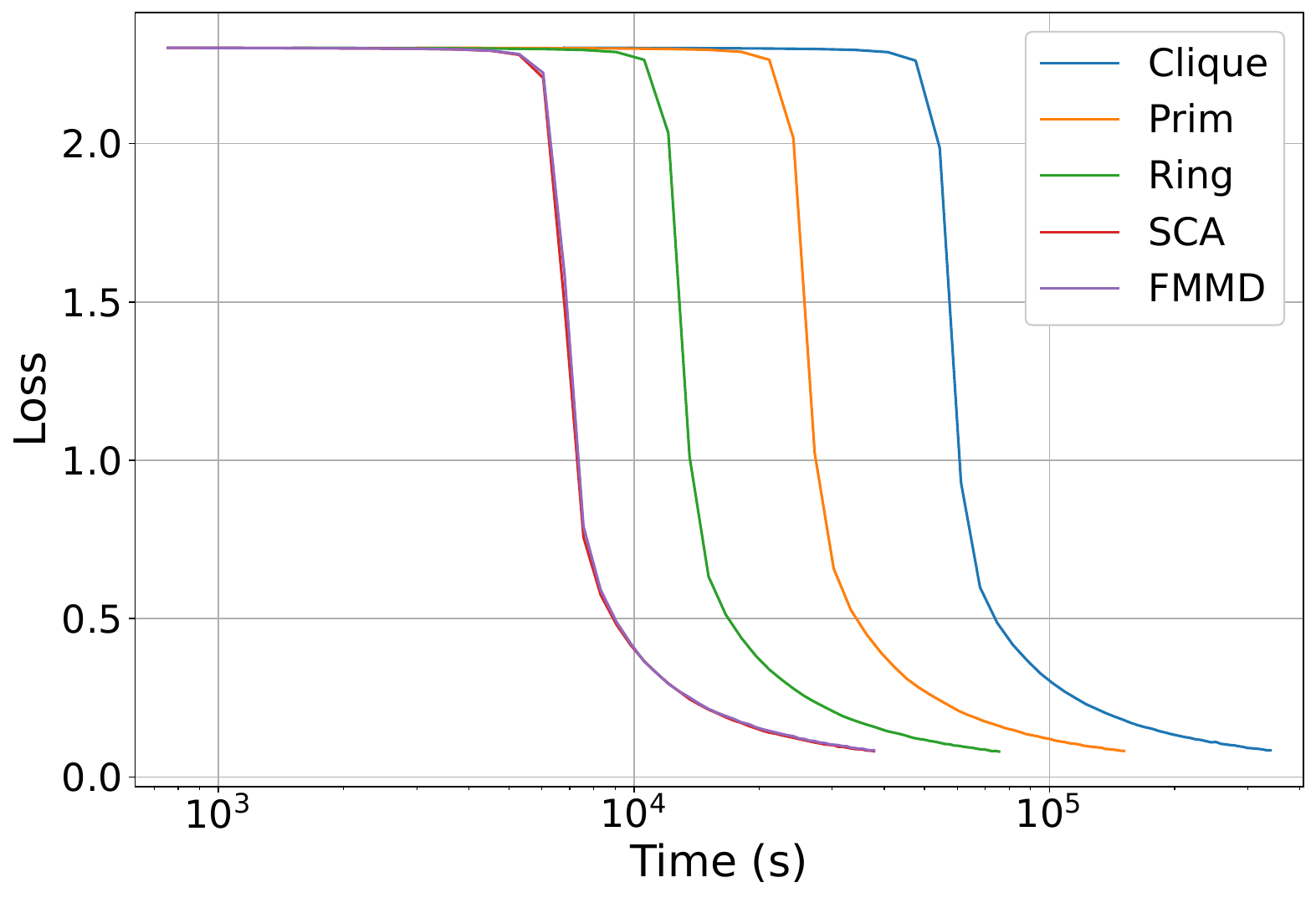}}
\vspace{-.1em}
\end{minipage}
\begin{minipage}{.495\linewidth}
\centerline{
\includegraphics[width=1\linewidth]{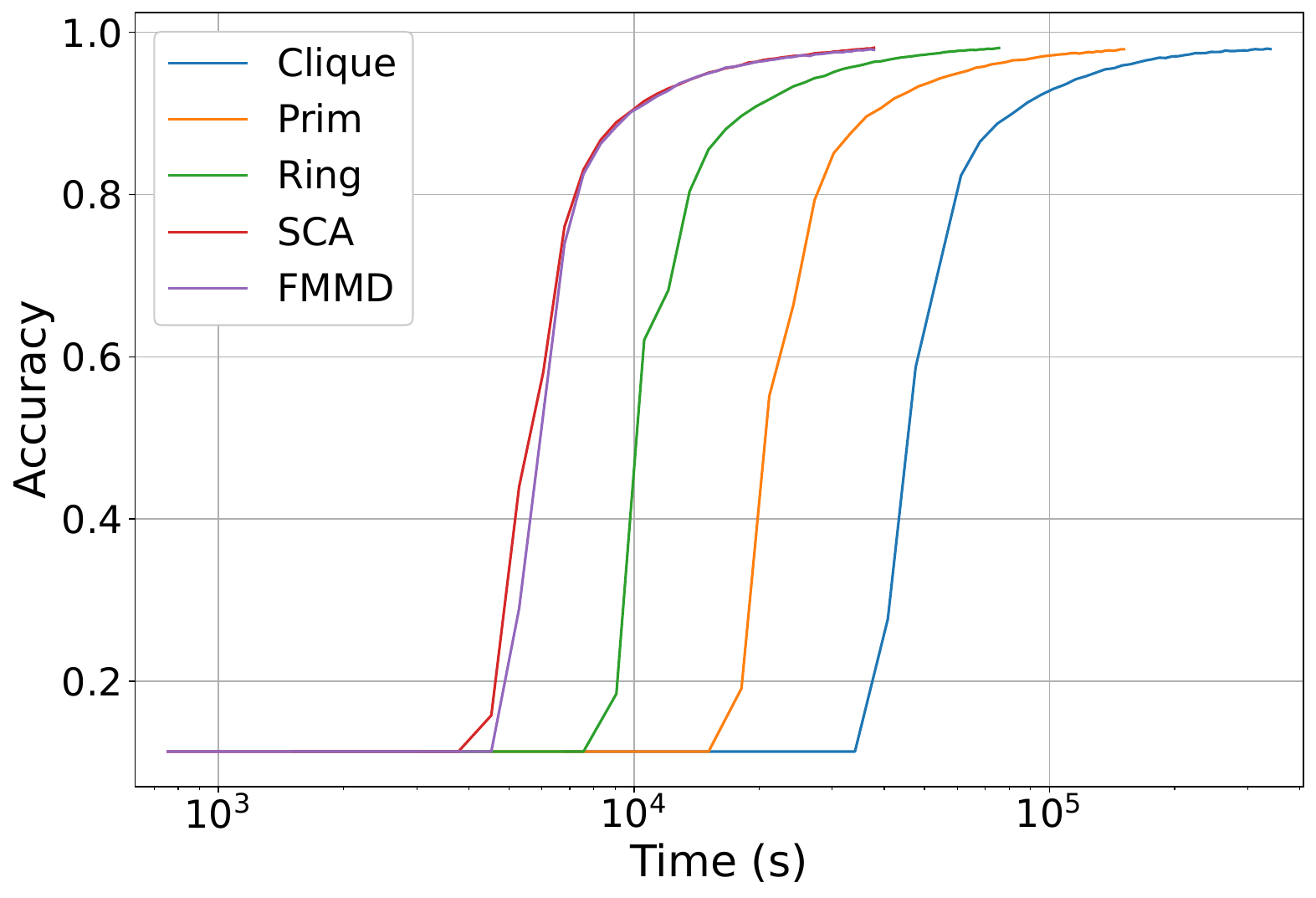}}
\vspace{-.1em}
\end{minipage}
\begin{minipage}{.495\linewidth}
\centerline{
\includegraphics[width=1\linewidth]{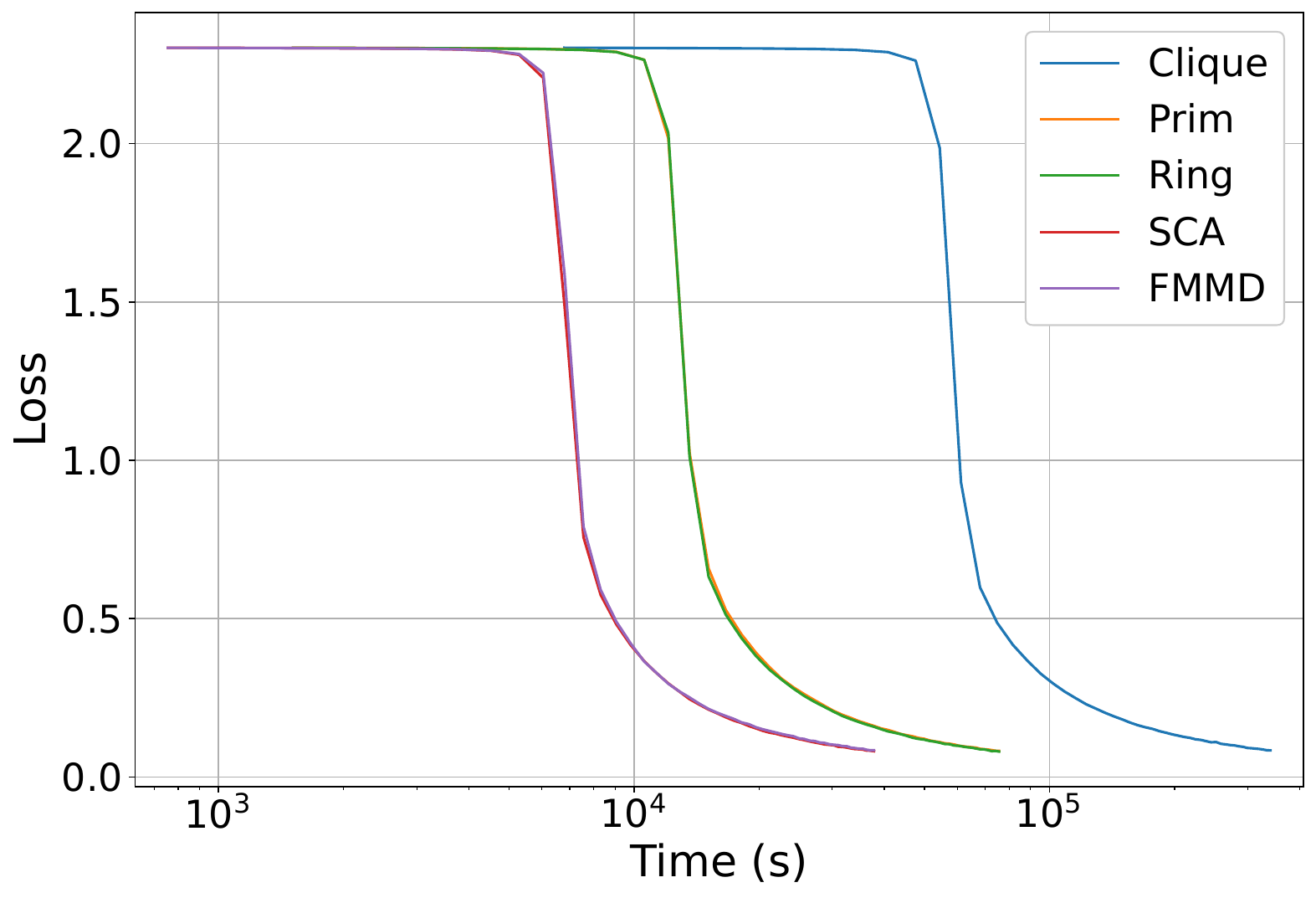}}
\vspace{-.1em}
\end{minipage}
\vspace{-.1em}
\begin{minipage}{.495\linewidth}
\centerline{
\includegraphics[width=1\linewidth]{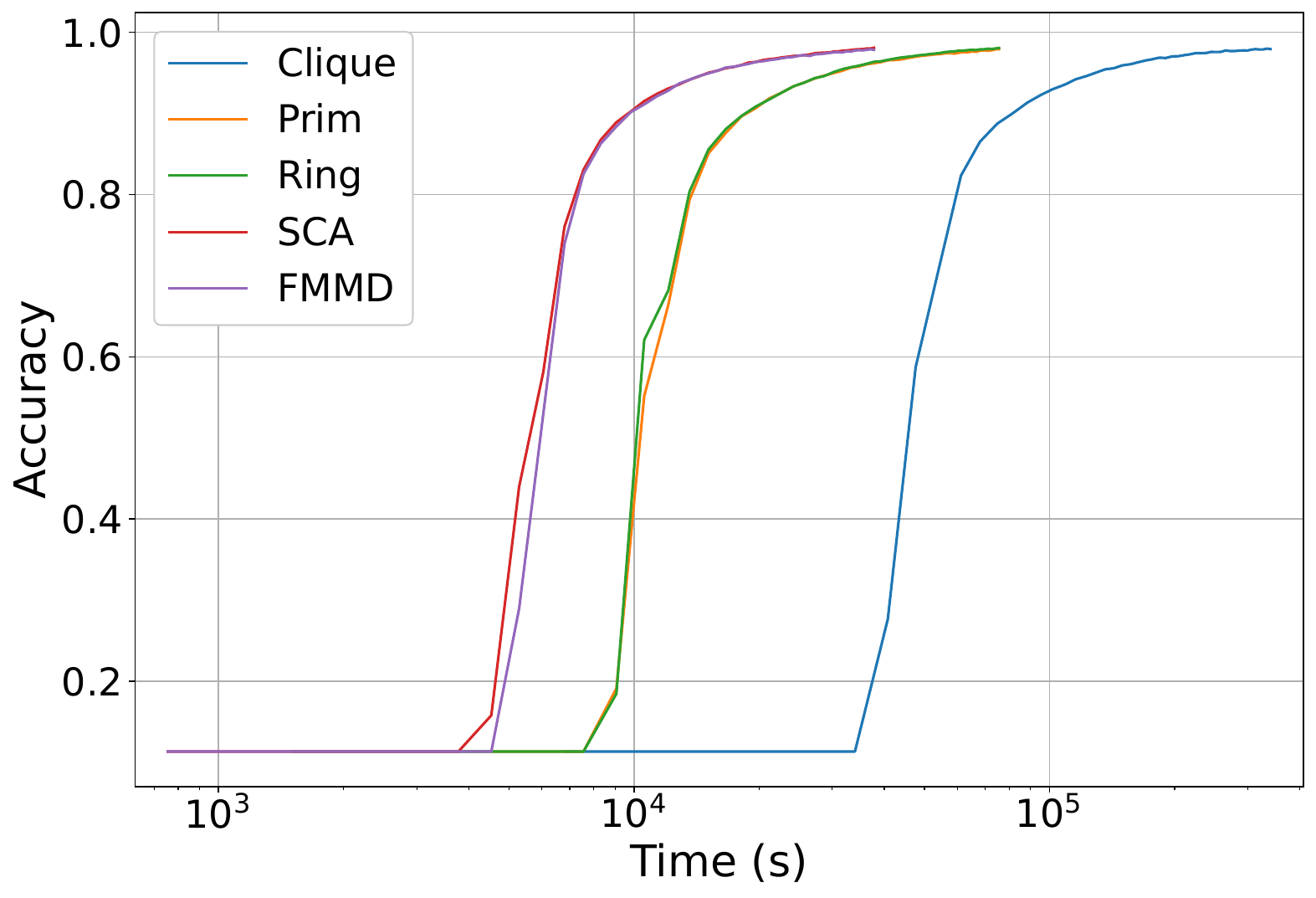}}
\vspace{-.1em}
\end{minipage}
\vspace{-1em}
\caption{Comparison with benchmarks on MNIST: first row - loss/accuracy over \#epochs, second row - loss/accuracy over $\overline{\tau}$, third row - loss/accuracy over $\tau$.  
} \label{fig:roofnet_no_error_mnist}
\vspace{-.0em}
\end{figure}

\begin{table}[t!]
\small
    \centering
    \begin{tabular}{c|c|c|c|c|c}
         &  SCA & FMMD & Prim & Ring & Clique \\
         \hline
routing by \eqref{eq:min-time} & 2.94  & 2.53 & 2.92 & 4.8205 & $>1000$\footnotemark\\       
routing by \eqref{eq:min-time simple} & 1.55  & 0.745 & 0.3 & 0.32 & 0.47
    \end{tabular}
    \vspace{.5em}
    \caption{Running times (s) for MNIST dataset (including link activation design, link weight design, and overlay routing). 
    }
    \label{tab: run_time_algs_MNIST}
    \vspace{-1em}
\end{table} 
\footnotetext{When using Gurobi to solve \eqref{eq:min-time}, the solver does not converge in 1,000 s.}

\fi

\end{document}